\documentclass[11pt]{article}
\usepackage[utf8]{inputenc}
\usepackage[twoside,letterpaper,margin=1in]{geometry}
\usepackage{microtype}
\usepackage{graphicx,url,amsmath,amsfonts,amssymb,amsthm,subfigure,bbm,bm}

\usepackage{hyperref}

\usepackage[dvips]{epsfig}
\usepackage{thm-restate}
\usepackage{float}
\usepackage{color}

\definecolor{blueblack}{rgb}{0,0,.7}

\newcommand{\emphdef}[1]{%
  \textcolor{blueblack}{%
    \textbf{\emph{#1}}%
  }%
}

\theoremstyle{plain}
\newtheorem{theorem}{Theorem}[section]
\newtheorem{proposition}[theorem]{Proposition}

\newtheorem{corollary}[theorem]{Corollary}
\newtheorem{lemma}[theorem]{Lemma}

\newcounter{sideremark}

\newcommand{\paren}[1]{\left({#1}\right)}

\def\eg{{\it e.g.,}~}

\newcommand{\eps}{\varepsilon}
\renewcommand{\epsilon}{\varepsilon}

\def\PP{\mathcal{P}}
\def\OPT{{\text{OPT}}}
\def\Sk{{Sk}}

\title{A Near-Linear Approximation Scheme for Multicuts\\ of Embedded
  Graphs with a Fixed Number of Terminals%
  \thanks{A previous version appeared in \emph{Proc.\ 29th ACM--SIAM Symp.\ on
      Discrete Algorithms} (SODA'18).  This work is supported by the French ANR Blanc projects
    ANR-12-BS02-005 (RDAM), ANR-16-CE40-0009-01 (GATO), and
    ANR-17-CE40-0033 (SoS).}%
}

\author{Vincent Cohen-Addad\thanks{
    Sorbonne Universit\'es, UPMC Univ Paris 06, CNRS, LIP6, Paris, France. 
    Email: \url{vcohenad@gmail.com} }%
  \and %
  \'Eric Colin de Verdi\`ere\thanks{LIGM, CNRS, Univ Gustave Eiffel, ESIEE
Paris, F-77454 Marne-la-Vall\'ee, France. Work partly done while the author was at
    D\'epartement d'Informatique, \'Ecole normale sup\'erieure, Paris.
    Email: \url{eric.colindeverdiere@u-pem.fr}.} %
  \and %
  Arnaud de Mesmay\thanks{LIGM, Univ. Gustave Eiffel, CNRS, ESIEE Paris, F-77454 Marne-la-Vall\'ee, France. Email:
\url{arnaud.demesmay@u-pem.fr}}
}

\begin{document}

\begin{titlepage}
  \maketitle
  \thispagestyle{empty}

\begin{abstract}
\normalsize
  For an undirected edge-weighted graph $G$ and a set $R$ of pairs of
  vertices called pairs of \textit{terminals}, a multicut is a set of
  edges such that removing these edges from $G$ disconnects each pair
  in $R$. We provide an algorithm computing a
  $(1+\varepsilon)$-approximation of the minimum multicut of a graph
  $G$ in time $(g+t)^{(O(g+t)^3)}\cdot(1/\varepsilon)^{O(g+t)} \cdot n
  \log n$, where $g$ is the genus of $G$ and $t$ is the number of
  terminals.

  This is tight in several aspects, as the minimum multicut problem is both
  APX-hard and W[1]-hard (parameterized by the number of terminals), even
  on planar graphs (equivalently, when $g=0$).

  Our result, in the field of fixed-parameter approximation algorithms,
  mostly relies on concepts borrowed from computational topology of graphs
  on surfaces.  In particular, we use and extend various recent techniques
  concerning homotopy, homology, and covering spaces.  Interestingly, such
  topological techniques seem necessary even for the planar case.  
  We also exploit classical ideas stemming from approximation
  schemes for planar graphs and low-dimensional geometric inputs. A key
  insight towards our result is a novel characterization of a minimum
  multicut as the union of some Steiner trees in the universal cover of the
  surface in which $G$ is embedded.
\end{abstract}

\end{titlepage}
\setcounter{page}{1}
\section{Introduction}
Cuts and flows are fundamental objects in combinatorial optimization and
have generated a very deep theory.  One classical problem is
\textsc{Multicut}: Let $G=(V,E)$ be an undirected graph, let $T$ be a
subset of vertices of~$G$, called \emphdef{terminals}, and let $R$ be a set
of unordered pairs of vertices in~$T$, called \emphdef{terminal pairs}.  A
subset $E'$ of~$E$ is a \emphdef{multicut} (with respect to~$(T,R)$) if for
every terminal pair $\{t_1,t_2\}\in R$, the vertices $t_1$ and~$t_2$ lie in
different connected components of the graph $(V,E\setminus E')$; see
Figure~\ref{F:illustration}(a).  Assuming $G$ is positively edge-weighted,
and given $G$, $T$, and $R$, the \textsc{Multicut} problem asks for a
multicut of minimum weight.

The case $|R|=1$ is the famous minimum cut problem and admits a
polynomial-time algorithm.  The case $|R|=2$ can also be solved in
polynomial time~\cite{Hu63}, but the case where $R$ is arbitrary is much
harder, and in fact even NP-hard to approximate within any constant factor
assuming the Unique Games Conjecture~\cite{ckkrs-hamsc-06}.  Thus, to
circumvent this hardness, a fruitful approach is to focus on specific
inputs that characterize the instances that are relevant in practice.

The study of cut and flow problems originates from practical problems
involving road or railway networks~\cite{s-hcot1-05}.  Since these can
usually be modeled by embedded graphs, there exists an important literature
towards the efficient computation of minimum cuts (or maximal flows) on
planar
graphs~\cite{bk-namsd-09,e-mfpsp-10,h-mfspn-81,hj-oamfu-85,kn-fpgsr-93,r-mscpu-83}
(among other references) or, more generally, graphs of bounded genus (which
can be embedded, i.e., drawn without crossings, on a
surface of bounded genus)~\cite{benw-apmcnr-16,cen-hfcc-12,cen-mcshc-09,efn-gmcse-12,en-mcsnc-11}.
However, \textsc{Multicut} remains APX-hard for planar
graphs~\cite{gvy-pdaifm-97}, and so to obtain an exact or
$(1+\eps)$-approximation algorithm it is needed to consider parameterized
algorithms.

A very natural parameter is the number $t=|T|$ of terminals; the case
$t=2$, corresponding to the minimum cut, 
has inspired a huge literature, and the problem is already NP-hard for $t=3$
in general~\cite{djpsy-cmc-94}. Unfortunately, \textsc{Multicut} is W[1]-hard when parameterized by the number of 
terminals, even for planar graphs~\cite{m-tlbpwc-12}.

Therefore, for planar graphs, the best possible result for
\textsc{Multicut} parameterized by the number of terminals is a
$(1+\eps)$-approximation.  We provide the first algorithm achieving such an
approximation guarantee, for planar and more generally for surface-embedded
graphs.  Moreover, the dependence in the size of the graph is near-linear
and the dependence in $\eps$ is polynomial:
\begin{theorem}\label{T:main}
  Let $G$ be an undirected, positively edge-weighted graph embeddable on a
  surface of genus $g$, orientable or not.  Let $n$ be the number of
  vertices and edges of~$G$.  Let $T$ be a set of $t$~terminals and $R$ be
  a set of unordered pairs of~$T$.  Then for every $\varepsilon >0$, we can
  compute a $(1+\varepsilon)$-approximation of the minimum multicut of~$G$
  with respect to $(T,R)$ in time $f(\varepsilon,g,t)\cdot n \log n$, where
  $f(\eps,g,t)=(g+t)^{O(g+t)^3}\cdot(1/\varepsilon)^{O(g+t)}$.
\end{theorem}

\begin{figure*}[ht]
  \begin{center}
    \def\svgwidth{\textwidth}
    \begingroup%
  \makeatletter%
  \providecommand\color[2][]{%
    \errmessage{(Inkscape) Color is used for the text in Inkscape, but the package 'color.sty' is not loaded}%
    \renewcommand\color[2][]{}%
  }%
  \providecommand\transparent[1]{%
    \errmessage{(Inkscape) Transparency is used (non-zero) for the text in Inkscape, but the package 'transparent.sty' is not loaded}%
    \renewcommand\transparent[1]{}%
  }%
  \providecommand\rotatebox[2]{#2}%
  \ifx\svgwidth\undefined%
    \setlength{\unitlength}{4818.28099861bp}%
    \ifx\svgscale\undefined%
      \relax%
    \else%
      \setlength{\unitlength}{\unitlength * \real{\svgscale}}%
    \fi%
  \else%
    \setlength{\unitlength}{\svgwidth}%
  \fi%
  \global\let\svgwidth\undefined%
  \global\let\svgscale\undefined%
  \makeatother%
  \begin{picture}(1,0.33246604)%
    \put(0,0){\includegraphics[width=\unitlength,page=1]{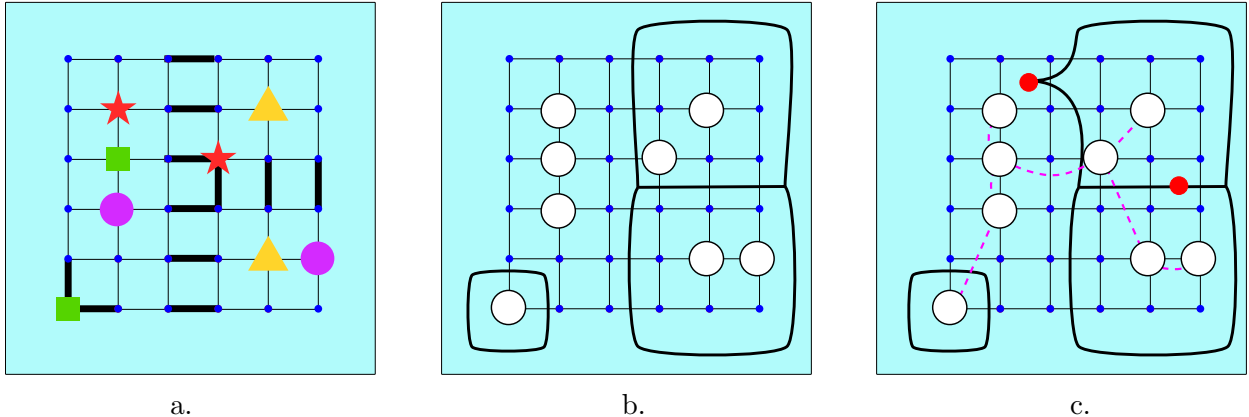}}%
    \put(0.14203605,0.00075721){\color[rgb]{0,0,0}\makebox(0,0)[b]{\smash{a.}}}%
    \put(0.50629305,0.00009404){\color[rgb]{0,0,0}\makebox(0,0)[b]{\smash{b.}}}%
    \put(0.8660031,0.00009404){\color[rgb]{0,0,0}\makebox(0,0)[b]{\smash{c.}}}%
  \end{picture}%
\endgroup%

    \caption{a. An instance of multicut (the polygons and stars denote pairs of terminals) and a solution. b. We consider the \em dual \em of the solution, and treat terminals as \em boundaries \em of the surface. c. The strategy for our algorithm is to compute a set of \emph{portals} (Section~\ref{S:portals}), whose removal cuts a near-optimal solution into trees and cycles.
The trees can then be computed (Section~\ref{S:alg}) as Steiner trees in the universal cover of the surface, which is obtained by
gluing a new copy of the surface each time we go through an edge of a system of arcs of the surface
(dashed lines---the system of arcs cuts the surface into a disk).
    }

    \label{F:illustration}
  \end{center}
\end{figure*}

\paragraph*{Comparison with existing hardness results}

As hinted at before, we argue that our result is, from various points
of view, the best possible (even if we were to restrict ourselves to planar
graphs, or equivalently assuming $g=0$).  Indeed, through its multiple parameters,
it finds a very intricate middle ground amidst the flurry of hardness
results that have been established for the \textsc{Multicut}
problem.
\medskip 
More precisely:
\begin{enumerate}
\item\textbf{Dependence on the number of terminals.} \textsc{Multicut} is
  known to be hard to approximate on planar graphs: Garg et
  al.~\cite{gvy-pdaifm-97} have proved that it is APX-hard on unweighted
  stars. Therefore, the exponential dependence on the number of terminals
  cannot be improved to a polynomial one, no matter the dependence on the
  genus, as this would yield a PTAS for planar graphs.

\item\textbf{Approximation factor.} \textsc{Multicut} is hard to solve
  exactly for planar graphs, even from the point of view of parameterized
  complexity: Marx~\cite{m-tlbpwc-12} has shown that it is W[1]-hard for
  planar graphs when parameterized by the number of terminals.  Actually,
  he has proved a lower bound of $n^{\Omega(\sqrt t)}$ assuming the
  Exponential Time Hypothesis.  Therefore,
  there is no exact FPT algorithm parameterized by $t$ unless W[1]=FPT (no matter the dependence in the genus).
  
\item\textbf{Dependence on the approximation factor.} The previous
  observation can be refined in the following way, using a parameterized
  version of the classical argument showing that ``integer-valued''
  strongly NP-complete problems do not admit an FPTAS. Indeed, since the reduction
  of Marx~\cite{m-tlbpwc-12} also applies to unweighted planar graphs, it
  also precludes an algorithm with running time
  $f(1/\varepsilon) h(t) p(n)$ for some computable function  $h$ and
  polynomials $f$ and $p$, which are even allowed to depend arbitrarily on $g$.
  Such an algorithm would yield an \textit{exact} FPT algorithm for
  \textsc{Multicut} on unweighted planar graphs by setting $\varepsilon=1/\Omega(n)$,
  since a solution to an unweighted \textsc{Multicut} instance has value
  $O(n)$.

\item\textbf{Dependence on the genus.} For general graphs,
  \textsc{Multicut} is known to be
  hard to approximate, even for fixed values of $t$: Dahlhaus et
  al.~\cite[Theorem~5]{djpsy-cmc-94} have proved that it is APX-hard for
  any fixed $t \geq 3$. Therefore (unless P=NP), the exponential dependence
  on the genus of the graph cannot be improved to a polynomial one, no
  matter the dependence on the number of terminals, as this would
  yield a $(1+\eps)$-approximation for general graphs (thus
  contradicting Dahlhaus et al.).

\end{enumerate}

\paragraph*{Comparison with existing algorithms}

Our result lies at the interface of parameterized complexity and
approximation algorithms---we refer to the article of Marx~\cite{m-pcaa-08}
for a survey of the results of this type. Within the language of that
survey, our algorithm is a fixed parameter tractable approximation scheme
(fpt-AS) with respect to its parameters (the genus, the number of
terminals, and the approximation factor).  Furthermore, we highlight that
the dependence on $1/\varepsilon$ is polynomial here. In particular, our
algorithm is a near-linear Fully Polynomial Time Approximation Scheme
(FPTAS) if the number of terminals and the genus are constant.  This is a
rare occurrence for approximation schemes: Indeed, as mentioned above, in
the unparameterized world, such an approximation scheme cannot exist for
``integer-valued'' strongly NP-complete problems. Thus, it is not clear how
to make use of the successful techniques designed for obtaining PTASes for
planar and embedded graphs, like for example ``brick decompositions''
(\eg~\cite{bhkm-ptaspm-12,bhm-assfpg-11,bdt-ptass-09,bkm-onasst-09}), as
they usually lead to algorithms that have exponential dependence in
$\varepsilon$.  Hence to obtain a polynomial dependence in $\eps$, we
develop a new, specific toolset.

The \textsc{Multicut} problem has been the subject of many exact or
approximate algorithms.  The most relevant to this article is a recent
algorithm~\cite{c-mpbgg-17} by the second author of this paper which solves
the exact version of \textsc{Multicut} in time
$O((g+t)^{O(g+t)}n^{O(\sqrt{g^2+gt})})$.  There also exist constant-factor
approximations to solve \textsc{Multicut} on graphs of bounded genus (or
excluding some fixed minor) using techniques based on padded decompositions
(see \eg~\cite{aggnt-crtspd-14,ft-idtgef-03,kpr-exndmf-93}), and Garg et
al.~\cite{gvy-amfmmc-96} have provided a $O(\log t)$-approximation for
general graphs. Regarding parameterized algorithms, a recent flurry of
results has culminated into two simultaneous
proofs~\cite{bdt-mfpt-11,m-fptmps-14} that unweighted \textsc{Multicut} is
FPT when parameterized by the \textit{cost of the solution}.  We emphasize
that this result is not comparable to ours due to the difference in the
choice of parameters.

A special instance of this problem that has been more studied is the
\textsc{Multiway cut} problem (also known as \textsc{Multiterminal cut}),
where the set~$R$ of pairs of terminals comprises all the pairs of vertices in
the set $T$ of terminals.  All the hardness results mentioned above for
\textsc{Multicut} were actually proved for \textsc{Multiway cut}, 
with one exception: The planar version of \textsc{Multiway cut} is not
APX-hard, and indeed Bateni et al.~\cite{bhkm-ptaspm-12} have provided
a PTAS for this planar problem. Regarding exact algorithms, Klein and
Marx~\cite{km-spktc-12} have shown how to solve
planar \textsc{Multiway cut} exactly in time $2^{O(t)}\cdot
n^{O(\sqrt{t})}$, coming very close to the lower bound of
$n^{\Omega(\sqrt t)}$ mentioned above. Chekuri and
Madan~\cite{cm-amdg-17} have recently investigated the complexity of
approximating \textsc{Multicut} with weaker constraints on the terminal
pairs, in both the undirected and the directed setting.

\paragraph*{Overview of the strategy and of the techniques}

The previous works on cuts and flows in embedded graphs have revealed
a strong relationship between these a priori purely combinatorial
constructs and some \textit{topological} aspects of the underlying
surface. For example, \textit{homology} is a key tool in the study of
minimum cuts of graphs on
surfaces~\cite{cen-mcshc-09,en-mcsnc-11}. Similarly, the starting
point of our algorithm is the observation that the \textsc{Multicut}
problem for surface-embedded graphs is inherently a topological
problem; this opens the possibility to use and enhance techniques from
computational topology of graphs on surfaces, which turns out to be 
the key to our approach.  We stress that these topological techniques are
essential, and are not a consequence of the fact that our result involves
surface-embedded graphs:
neither the algorithm nor its proof of correctness gets significantly
simpler by considering only planar graphs.

Indeed, in the same way that minimal cycles are dual to minimal cuts
in planar graphs, the edges dual to a minimum multicut form in some
sense the ``shortest dual graph'' that separates topologically the
required sets of terminals; see Figure~\ref{F:illustration}(b) and
Section~\ref{S:overview}.  Henceforth, we always use this dual
reformulation and can completely forget about the original problem.
While this topological point of view already underlies the early
algorithms for planar \textsc{Multiway Cut}~\cite{djpsy-cmc-94}, it
was only recently made precise by the second author of this
article~\cite{c-mpbgg-17}, based on the framework of
\textit{cross-metric surfaces}~\cite{ce-tnpcs-10}.  He also proved
that the topology of such an optimal \textit{multicut dual} can be
quantified appropriately by bounding its number of intersections with
a certain \emph{cut graph} of the surface (equivalently in this paper,
a \emph{system of arcs}---in the planar case, a certain tree spanning
the terminals).  The algorithm of~\cite{c-mpbgg-17} for solving
\textsc{Multicut} exactly ``guesses'' (by enumerating all possibilities)
the topology (i.e.,
the \textit{homotopy class}) of the edges of the optimal multicut
together with the position of the vertices. In particular, an easy
instance of \textsc{Multicut} is one where all the components of an
optimal multicut dual are simple cycles. In this case, their topology
can be guessed and they can all be computed separately as shortest
homotopic cycles. Therefore, the bulk of our effort in this paper
concerns dealing with the graphical components of the dual solution,
and we will disregard cyclic components in the rest of this overview.

A new topological idea driving our approach is the following:
Suppose that we have managed to represent a multicut dual as a union 
of trees (see Figure~\ref{F:illustration}(c)).  Replacing, in the
multicut, such a tree with a shortest tree with the same topology
(roughly, homotopy) and the same location for the leaves preserves the
multicut property, and of course does not increase the length.  (This
may seem intuitively clear, but a great care is needed to make this
argument rigorous.)  Knowing only the topology of the tree and the
location of the leaves, we show that one can compute such a shortest
tree in near-linear time parameterized by the genus and the number of
terminals (Section~\ref{S:alg}); this subroutine combines topological
ideas to compute shortest homotopic curves with universal cover
constructions~\cite{ce-tnpcs-10} with the Steiner tree algorithm by
Dreyfus and Wagner~\cite{dw-spg-71,emv-ssmmc-87} that is
fixed-parameter tractable in the number of terminals.

We then mix these topological tools with standard techniques
stemming from the
design of approximation schemes in low-dimensional
geometric inputs (originating with Arora~\cite{a-ptaset-98} and
Mitchell~\cite{m-gsaps-99}) which have been successfully adapted to planar graphs, see for example~\cite{bhkm-ptaspm-12,bkm-onasst-09}.  As hinted above, the number of possible
topologies for a tree-like portion of an optimal solution is bounded,
and thus can be ``guessed'', so all the difficulty resides in
computing a small set of points on the surface, called \emph{portals},
whose removal cuts a near-optimal solution into trees (see
Figure~\ref{F:illustration}(c) and Section~\ref{S:portals}).  This, in
turn, boils down to computing a \emph{skeleton}, a graph of controlled
length with respect to the length of the optimal solution that cuts a
near-optimal solution into a forest.  The portals can then be placed
at regular intervals along the skeleton, to ensure that the
near-optimal solution, after creating short detours, passes through
these portals.

Thus, the only remaining difficulty is to compute the skeleton efficiently.
This is actually the most technical part of this article
(Section~\ref{S:skeleton}), also extensively relying on many topological
tools, in particular certain covering spaces.  For this approach to work,
we (roughly) need to treat the ``long'' cycles of the solution separately
(Section~\ref{S:nolongcycle}), and additional technicalities are needed to
obtain a small dependence in the genus and the number of terminals
(Section~\ref{S:exhaustive}).  We give a more detailed overview of the
construction of the skeleton in Section~\ref{S:overview2}, once we have
introduced the necessary terminology.

\paragraph*{General remarks}

In summary, while we use some ideas from PTAS in computational geometry and
planar graphs, the vast majority of the arguments, and the most delicate
ones, are of topological nature.  We note that, throughout this article, we
use several techniques which are somewhat at odds with the typical setting
encountered in computational topology for graphs on surfaces.  For example,
computations in some covering space, when projected back to the surface,
may involve undesired crossings; and it is more natural to have a
near-optimal solution take detours after the placement of portals if it is
allowed to cross itself. This is why many of our arguments deal with
\textit{drawings} of graphs instead, which allow a graph to intersect
itself. One of the perks of our approach is that since being a multicut
dual is a homotopical property, most of our techniques dovetail naturally
with this change of perspective, and the setting of cross-metric surfaces
also allows for a seamless algorithmic treatment of drawings of graphs.

\section{Preliminaries}
\label{S:prelim}
\subsection{Graphs, surfaces, and homotopy}

We recall here standard definitions on the topology of surfaces.  For
general background on topology, see for example Stillwell~\cite{s-ctcgt-93}
or Hatcher~\cite{h-at-02}.  For more specific background on the topology of
surfaces in our context, see recent articles and surveys on the same
topic~\cite{ce-tnpcs-10,c-ctgs-17}.

Let $S$ be a (compact, connected) \emphdef{surface} without boundary.  If
$S$ is an orientable surface, its \emphdef{genus} $g\geq0$ is even, and $S$
is (homeomorphic to) a sphere with $g/2$ handles attached.  If $S$ is a
non-orientable surface, its \emphdef{genus} is a positive integer~$g$, and
$S$ is a sphere with $g\geq1$ disjoint disks replaced by M\"obius strips.
A \emphdef{surface with boundary} is obtained from a surface without
boundary by removing a set of open disks with disjoint closures.  The boundary
of such a disk is a \emphdef{boundary component} of~$S$.

A \emphdef{path} is a continuous map from $[0,1]$ to~$S$.  Two paths are
\emphdef{homotopic} if there is a continuous deformation between them on
the surface that keeps the endpoints fixed.  An \emphdef{arc} is a path
with endpoints on the boundary of the surface.  A \emphdef{closed curve} is
a continuous map from the unit circle to~$S$.  Two closed curves are
(freely) \emphdef{homotopic} if there is a continuous deformation between
them.  A path or closed curve is \emphdef{simple} if it is one-to-one. 

A closed curve~$\gamma$ is \emphdef{contractible} if it is homotopic to a closed curve that is a constant map.  If $\gamma$ is simple, a small neighborhood of the
image of~$\gamma$ is either an annulus or a M\"obius strip; in the former
case, $\gamma$ is \emphdef{two-sided}, and otherwise it is
\emphdef{one-sided}.  Of course, one-sided curves exist only if the surface
is non-orientable\footnote{The reader not familiar with topology may safely
  skip all considerations involving non-orientable surfaces and one-sided
  curves.  If the input graph is planar, then one can even restrict oneself
  to a surface that is a sphere with disjoint disks removed.}.

We consider \emphdef{drawings} of graphs $C$ on the surface~$S$, which
are assumed to be in general position: There are finitely many
(self-)intersection points of the drawing, each involving exactly two
edges, and these two edges actually cross at that point, in their
relative interior. A similar definition holds for closed curves, and
two graphs drawn in~$S$ are in general position with respect to each other if their
union is in general position.  We frequently
abuse terminology and identify the abstract graph~$C$ with its drawing
on~$S$.  An \emphdef{embedding} of the graph~$C$ is an injective
(``crossing-free'') drawing.  An embedding is \emphdef{cellular} if
its faces are disks.

We say that $C$ is an \emphdef{even graph} if every vertex has even degree.
We will work with \emphdef{homology} over the field $\mathbb{Z}_2$, which
is an equivalence relation between even graphs that is coarser than
homotopy.  We can sketch the definition here: Let $C_1$ and~$C_2$ be two
even graphs embedded on~$S$, crossing finitely many times.  Then the
closure of the symmetric difference of the images of~$C_1$ and~$C_2$,
denoted by $C_1+C_2$, is also the image of an even subgraph embedded
on~$S$; we say that $C_1$ and~$C_2$ are \emphdef{homologous} if the faces
of $C_1+C_2$ can be colored with two colors so that all the boundary
components of~$S$ have the same color and the boundary between the two colors
is precisely the image of~$C_1+C_2$.
An even graph is \emphdef{null-homologous} if it is homologous to a contractible closed curve. We refer to, e.g.,~\cite{cen-mcshc-09,c-ctgs-17} for more background on homology.

Let $\Gamma$ be a set of closed curves in general position on a
surface~$S$.  A \emphdef{monogon} is a disk on~$S$ whose interior does
not meet~$\Gamma$ and whose boundary is formed by a subpath of a curve
in~$\Gamma$.  A \emphdef{bigon} on~$S$ is a disk on~$S$ whose interior
does not meet~$\Gamma$ and whose boundary is formed by the
concatenation of two subpaths of curves in~$\gamma$ (possibly the same
curve) such that these two subpaths are disjoint except at their
endpoints.

If $C$ is drawn on~$S$ with crossings, there is an ambiguity when
talking about \emphdef{cycles} of $C$. We adopt the convention that a
\emphdef{cycle} in $C$ is a cycle in $C$ in the graph-theoretical
sense (a closed walk, not reduced to a single vertex, without vertex
repetition), or by abuse of language the image of that cycle. Thus a
cycle in $C$ is a closed curve on~$S$ (which may cross itself). 

The input of our algorithm is a graph that can be embedded on a surface,
orientable or not, with genus~$g$.  Every graph can be embedded on some
surface of sufficiently large genus.  Actually, there is an algorithm that takes as
input a graph~$G$ with $n$ vertices and edges, and a surface $S$ specified
by its genus~$g$ and by whether it is orientable, and decides in
$2^{\text{poly}(g)}\cdot n$ whether $G$ embeds on~$S$.  Moreover, in the affirmative, it
computes a cellular embedding of~$G$ on a surface $S'$ with the same
orientability as~$S$ and with genus at most~$g$~\cite{kmr-sltae-08}.

It is useful to consider the surface with boundary obtained by removing
small disks around each terminal; this operation transforms every terminal
vertex~$v$ of~$G$ into $\deg(v)$ one-degree vertices on the boundary
of~$S$.  Henceforth, {\color{blueblack}$\bm{S}$} is a surface with
genus~{\color{blueblack}$\bm g$} and~{\color{blueblack}$\bm t$} boundary
components.  Each boundary component of~$S$ corresponds to a terminal, and
we can thus view the set of terminal pairs {\color{blueblack}$\bm R$} as a
set of pairs of boundary components to be separated.

\subsection{Cross-metric surfaces, systems of arcs, and topology}\label{S:topology}

A \emphdef{cross-metric surface} $(S',G')$ is a surface~$S'$ (possibly with
boundary) equipped with a positively edge-weighted graph~$G'$ cellularly
embedded on~$S'$~\cite[Section~1.2]{ce-tnpcs-10}.  Curves (paths and closed
curves) on $(S',G')$ are assumed to be in general position with~$G'$.  The
\emphdef{length} of a curve~$c$, sometimes denoted by~$\color{blueblack}{\bm{|c|}}$, is the sum of the weights of the edges of~$G'$
crossed by the curve, counted with multiplicity.  This provides a discrete
notion of metric for~$S'$.  In this paper, the cross-metric surface we
frequently use is $(S,G)$, as defined in the previous paragraph. In
particular, $(S,G)$ always has at least one boundary component. When we
use terms such as ``shortest'' or ``distance'', this implicitly refers to
this notion of distance.

As in previous papers~\cite[Proposition~4.1]{c-mpbgg-17} (see also, e.g.,
Chambers et al.~\cite{ccelw-scsh-08} and Erickson and Nayyeri~\cite{en-mcsnc-11}), we will use an algorithm to compute
a \emphdef{greedy system of arcs}~$K$:
\begin{proposition}\label{P:cutgraph}
  In $O(n\log n+(g+t)n)$ time, we can compute a set~$K$ of $O(g+t)$
  disjoint, simple arcs on~$S$ in general position with respect to~$G$,
  such that $S\setminus K$ is a disk.  Moreover, each arc is the
  concatenation of two shortest paths on~$S$, and is a shortest homotopic
  path on~$S$.
\end{proposition}

We note that the presentation in~\cite{c-mpbgg-17} is slightly different
because the construction takes place in the surface without boundary, but
the result is the same.  When we write that the arc is the concatenation of
two shortest paths $uv$ and~$vw$, here it could be that the point~$v$ lies
on an edge of~$G$, not taken into account for computing the lengths of
paths $uv$ and~$vw$.  Also, the fact that each arc is a shortest homotopic
path is folklore and follows, e.g., from the analysis in~\cite{c-scgsp-10}
and the fact that two homotopic paths are homologous.

As an illustrating special case for the above proposition, assume that $S$
has genus zero (a sphere with disjoint open disks removed).  If we contract
each boundary component to a point, then $K$ is a tree spanning
the terminals.  Thus, the set $K$ can be computed using, e.g., a minimum
spanning tree algorithm, slightly adapted to the cross-metric setting and
to accommodate for the boundary components instead of the terminals.

Henceforth, we fix one greedy system of arcs $K$ for $(S,G)$ that will be used throughout the paper. All graphs are drawn in general position with respect to
$K\cup G$, and the \emphdef{complexity} of a graph drawn on $S$ is the sum of its number of vertices, edges, and the number of intersections with $G$ (this coincides with the sum of the lengths of the edges if $G$ is unweighted). In order to avoid issues with non-uniqueness of shortest paths, we will use throughout the article the rule that whenever we compute a shortest (homotopic) path, out of all the possible paths having the same length, we output one with the fewest crossings with $K$.

Let $C$ be a graph embedded on~$S$ and in general position with~$K$.  Since
$K$ cuts~$S$ into a disk, we can represent~$C$ by its image in the disk.
More precisely, we define the \emphdef{topology} of~$C$ to be the
information of the combinatorial map of the overlay of~$C$ and~$K$.

Finally, one notation: given a closed curve $\gamma$, we denote by
{\color{blueblack}$\bm{\gamma^o}$} a shortest closed curve in its
homotopy class.

\subsection{Covering spaces and shortest homotopic curves}

Our algorithm relies on
arguments used to compute shortest homotopic paths and
closed curves on the surface with boundary~$S$~\cite[Sections 1
and~6]{ce-tnpcs-10}. In a nutshell, the main idea is that computing shortest homotopic paths or shortest homotopic closed curves generally boils down to computing shortest paths in some region of controlled size (the \emph{relevant region}) within some space larger than the surface (a \emph{covering space}). Since our setting is 
simpler than the one in~\cite{ce-tnpcs-10}
(because $S$ has boundaries and and we not need the best
possible complexity with respect to $g$ and~$t$), 
we find it
worthwhile to summarize and explain these tools in a way that is tailored
to our purpose; see Appendix~\ref{A:covering}. 
The reader
unfamiliar with covering spaces will also find the relevant definitions
there.

\section{Multicut duals and overview}
\label{S:overview}
\subsection{Multicut duals}

A \emphdef{multicut dual}~$C$ is a graph drawn (possibly with crossings)
on~$S$ such that every path in $S$ connecting two boundaries of~$S$
corresponding to a terminal pair $\{t_1,t_2\}\in R$ meets the image
of~$C$. Thus, the set of edges crossed by~$C$ forms a multicut. Conversely,
any multicut corresponds to a multicut dual of the same cost; thus, it
suffices to compute a shortest multicut dual (see
also~\cite[Proposition~3.1]{c-mpbgg-17}). Henceforth, we focus on the
problem of computing a shortest multicut dual. Throughout this article, $K$ is a greedy system of arcs as defined and
computed in Proposition~\ref{P:cutgraph}.

Most of the structures we consider have small complexity, in the
following sense.
A graph~$C$ drawn on~$S$ is \emphdef{small} if it satisfies the
following conditions:
\begin{enumerate}\renewcommand{\theenumi}{(\roman{enumi})}
\item\label{sm:genpos} $C$ is drawn on~$S$ in general position with respect to $K \cup G$.
\item\label{sm:nodeg01} Each vertex of~$C$ has degree at least two.
\item\label{sm:compl} $C$ has $O(g+t)$ vertices and edges.
\item\label{sm:cross} The number of crossings between $C$ and each arc of~$K$ is $O(g+t)$.
\end{enumerate}
A graph $C$ drawn on~$S$ is \emphdef{eligible} if it satisfies the following conditions:
\begin{enumerate}
\item $C$ is small.
\item $C$ is embedded on~$S$.
\item Each face of $C$ contains at least one terminal.
\end{enumerate}

To determine whether an embedded graph is eligible, the information of its
topology is sufficient.  Thus, we also refer to the topology of any eligible graph as an eligible topology. The two following lemmas show that a shortest multicut dual has an eligible topology, and that the eligible topologies can be enumerated quickly.

\begin{lemma}\label{L:eligibility-struct}
  In $(g+t)^{O(g+t)}$ time, we can enumerate a set of $(g+t)^{O(g+t)}$
  eligible topologies, one of which is the topology of some shortest
  multicut dual.
\end{lemma}
\begin{proof}
  This follows from~\cite[Section~5, Proposition~5.1, and
  Proposition~6.1]{c-mpbgg-17}.  The algorithm is described
  in~\cite[Proposition~6.1]{c-mpbgg-17}.  It outputs eligible topologies
  (satisfying an additional condition that is irrelevant here), except that
  the last condition that each face of~$C$ contains at least one terminal
  is not necessarily enforced.  We thus discard the topologies that do not
  satisfy this condition.

  One of the topologies output by~\cite[Proposition~6.1]{c-mpbgg-17} is the
  topology of some shortest multicut dual.  As it turns out, in the
  algorithm, if the topology of a multicut dual~$C$ is produced, the
  topology of any multicut dual that is a subgraph of~$C$ is also produced,
  and in particular the topology of an inclusionwise minimal subgraph
  of~$C$ that is a multicut dual.  Each face of such an inclusionwise
  minimal subgraph of~$C$ contains at least one terminal.  Thus the
  algorithm produces an eligible topology of some shortest multicut dual.
\end{proof}

We let ${\color{blueblack}\bm{C_\text{\bf OPT}}}$ be a shortest multicut
dual that is eligible (which exists per Lemma~\ref{L:eligibility-struct}), and let
${\color{blueblack}\text{\bf OPT}}$ be its length.

\subsection{Multicut dual property}

We now introduce a sufficient condition for a graph drawn on~$S$ to be
a multicut dual, which will be used throughout this paper.  Let $C$ be
a graph drawn on~$S$.  

We remark that the property of being a multicut dual is not preserved by
homological transformations, contrary to, e.g., the dual of a minimum
cut~\cite{cen-mcshc-09}.  For example, a shortest multicut dual needs not
be an even graph; moreover, a pair of terminals is allowed to be separated by a
positive even number of edges.  However, homology can still be used in this
realm, by also looking at subgraphs of the multicut dual:
We say that $C$ has the \emphdef{multicut dual property} if, for every even
subgraph~$C'_\OPT$ of~$C_\OPT$, there is a subgraph~$C'$ of~$C$ such
that the homology classes of $C'$ and $C'_\OPT$ in the surface~$S$ are
equal.

Obviously, $C_\OPT$ has the multicut dual property.  It may appear strange
that this definition depends on the unknown graph~$C_\OPT$; in particular,
the multicut dual property is not easily checkable.  Roughly speaking,
starting from the unknown graph~$C_\OPT$, we will prove that one can modify
it, while maintaining the multicut dual property, into another graph with
more structural properties.  The existence of this new graph will be useful
for the algorithm.

\begin{lemma}\label{L:homology}
  Let $C$ be a graph drawn on~$S$.  Assume that $C$ has the multicut dual
  property.  Then $C$ is a multicut dual.
\end{lemma}
Note, however, that some multicut duals might not have the multicut dual
property.

\begin{proof}
  (In this proof, we are actually not using the fact that $C_\OPT$ is
  optimal, only that it is a multicut dual that is embedded.) %
  Let $(t_1,t_2) \in R$ be a terminal pair. Since $C_\OPT$ is a multicut
  dual, every path connecting $t_1$ and $t_2$ on $S$ meets $C_\OPT$.  Let
  $C'_\OPT$ be an inclusionwise minimal subgraph of~$C_\OPT$ that separates
  $t_1$ from~$t_2$.  This subgraph has exactly two faces, and every edge is
  incident to each of these two faces; thus, it is an even subgraph
  of~$C'_\OPT$, and there is a subgraph $C'$ of $C$ homologous to $C'_\OPT$
  in $S$, i.e., $C'+C'_\OPT = U$, where $U$ is homologically trivial. Now,
  any path connecting $t_1$ and $t_2$ in $G$ crosses $U$ an even number of
  times (because homology is considered in~$S$), and it crosses $C'_\OPT$
  an odd number of times (because $C'_\OPT$ has two faces, each containing
  one of $t_1$ and~$t_2$). Therefore, each such path crosses $C'$ an odd
  number of times, i.e., at least once.  Since this applies for any
  terminal pair, this shows that $C$ is a multicut dual.
\end{proof}

\subsection{Good multicut duals}

We now define good multicut duals, which are a central concept of
this paper.  A graph~$C$ drawn on~$S$ is a \emphdef{good multicut dual} if,
as an abstract graph, it is the disjoint union of subgraphs
$(C_0,C_1,\ldots,C_k)$ (whose images on~$S$ may have crossings and
self-crossings), satisfying the following conditions.

\newcommand\expl[1]{\textbf{[#1]}}
\begin{enumerate} 
\item\label{c:small}\expl{small} Each of $C_0,C_1,\ldots,C_k$ is small.
\item\label{c:cycles}\expl{structure} $C_1,\ldots, C_k$, viewed as abstract
  graphs, are cycles.  ($C_0$ is arbitrary, and in particular may be non-connected.)
\item\label{c:nbcycles}\expl{number of cycles} $k=O(g+t)$.
\item\label{c:length}\expl{length} The length of~$C$ is at most
  $(1+O(\varepsilon))\OPT$.
\item\label{c:homology}\expl{multicut dual property} $C$ has the multicut
  dual property.
\end{enumerate}

The $C_0$ component of a good multicut dual will be called its
\emphdef{core}.

\begin{figure}
  \centering
\def\svgwidth{14cm}
\begingroup%
  \makeatletter%
  \providecommand\color[2][]{%
    \errmessage{(Inkscape) Color is used for the text in Inkscape, but the package 'color.sty' is not loaded}%
    \renewcommand\color[2][]{}%
  }%
  \providecommand\transparent[1]{%
    \errmessage{(Inkscape) Transparency is used (non-zero) for the text in Inkscape, but the package 'transparent.sty' is not loaded}%
    \renewcommand\transparent[1]{}%
  }%
  \providecommand\rotatebox[2]{#2}%
  \newcommand*\fsize{\dimexpr\f@size pt\relax}%
  \newcommand*\lineheight[1]{\fontsize{\fsize}{#1\fsize}\selectfont}%
  \ifx\svgwidth\undefined%
    \setlength{\unitlength}{3264.08603865bp}%
    \ifx\svgscale\undefined%
      \relax%
    \else%
      \setlength{\unitlength}{\unitlength * \real{\svgscale}}%
    \fi%
  \else%
    \setlength{\unitlength}{\svgwidth}%
  \fi%
  \global\let\svgwidth\undefined%
  \global\let\svgscale\undefined%
  \makeatother%
  \begin{picture}(1,0.46298685)%
    \lineheight{1}%
    \setlength\tabcolsep{0pt}%
    \put(0,0){\includegraphics[width=\unitlength,page=1]{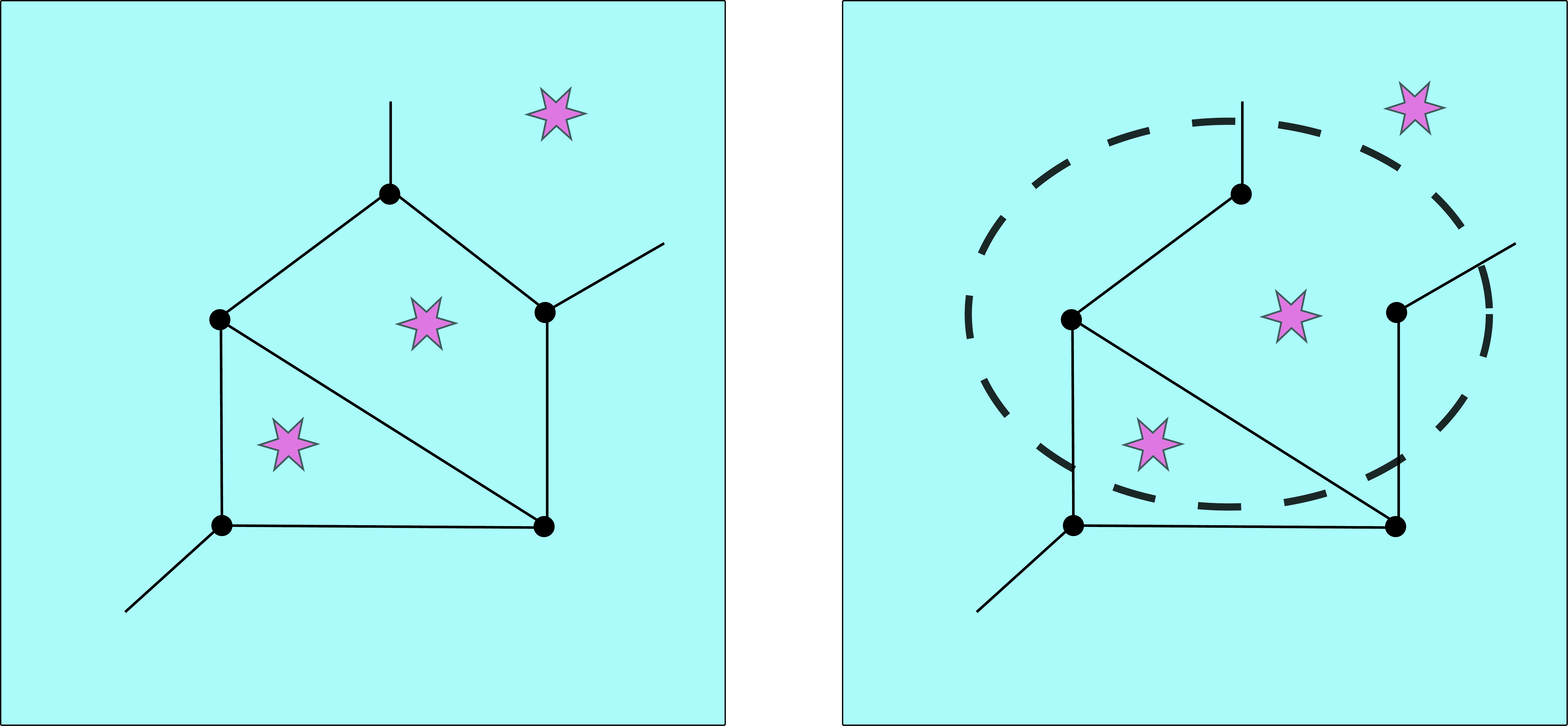}}%
    \put(0.07912247,0.04955894){\color[rgb]{0,0,0}\makebox(0,0)[lt]{\lineheight{0}\smash{\begin{tabular}[t]{l}$C$\end{tabular}}}}%
    \put(0.23975177,0.14240139){\color[rgb]{0,0,0}\makebox(0,0)[t]{\lineheight{0}\smash{\begin{tabular}[t]{c}$\gamma$\end{tabular}}}}%
    \put(0.30800169,0.31323151){\color[rgb]{0,0,0}\makebox(0,0)[t]{\lineheight{0}\smash{\begin{tabular}[t]{c}$e$\end{tabular}}}}%
    \put(0.61389501,0.04120311){\color[rgb]{0,0,0}\makebox(0,0)[lt]{\lineheight{0}\smash{\begin{tabular}[t]{l}$C_0$\end{tabular}}}}%
    \put(0.71046301,0.38936232){\color[rgb]{0,0,0}\makebox(0,0)[t]{\lineheight{0}\smash{\begin{tabular}[t]{c}$C_1$\end{tabular}}}}%
  \end{picture}%
\endgroup%

\caption{Portions of good multicut duals.  Stars denote terminals. Both
  sides depict good multicut duals. On the left, it consists only of its core $C_0$, while on the right, there are two components $C_0$ and $C_1$.}
  \label{F:examplegood}
\end{figure}

 The rationale behind the a priori strange splitting in this
 definition into a core $C_0$ and a family of
 cycles $C_1, \ldots C_k$ will be explained in the next
 subsection. Figure~\ref{F:examplegood} pictures a portion of two good multicut
 duals, one consisting only of a core $C_0$ on the left, and
 one also using a cycle $C_1$ on the right: in both cases the three
 stars are separated. We remark:
\begin{itemize}
\item The definition of a good multicut dual depends on a specific choice
  of a greedy system of arcs $K$ and of a specific choice of an optimum
  multicut dual~$C_\OPT$, but these are fixed throughout this paper (again,
  $C_\OPT$ is not known, but this is not a problem).
\item By the definition of~$C_\OPT$, the multicut dual $\{C_\OPT\}$ (where
  $C_\OPT$ is the core, and there is no cycle) is a good multicut dual with
  an eligible core.
\item We will care in particular about good multicut duals whose cores are
  eligible, namely, which satisfy the additional property: The core is
  embedded, and each face of the core contains at least one terminal.
  However, in the course of the proof, we will work with good multicut
  duals with a core that is not eligible.
\item Condition~\ref{c:homology}, together with Lemma~\ref{L:homology},
  implies that a good multicut dual is a multicut dual.
\item This definition uses the $O(\cdot)$ notation.  Each time the notion
  of good multicut dual is used, the constants hidden in the $O(\cdot)$
  notation are universal, independent from the input of the algorithm, but
  the constants will vary across statements.  More precisely, in some of
  the following sections, the main proposition
  (Propositions~\ref{P:shortcycles}, \ref{P:skeleton}, and~\ref{P:portals})
  involves some good multicut dual, and the constants in the $O(\cdot)$
  notation implicitly increase each time.  Since the number of increases is
  bounded, there is no danger in using this notation.
\item In the rest of the paper, we focus on computing a good multicut dual.
  Assuming we computed one, by Condition~\ref{c:length}, we can obtain a
  $(1+O(\varepsilon))$-approximation of~$\OPT$, by taking the edges in~$G$
  that are crossed by the good multicut dual.
\end{itemize}

\subsection{Overview for the construction of the
  skeleton}\label{S:overview2}

Recall that our main strategy is to compute a short skeleton that cuts the core of a
near-optimal multicut dual into a forest.  Indeed, one can then guess the
topologies of the trees and these locations for the leaves (via the
construction of the portals).  For each such guess, we determine the
shortest trees satisfying these properties by computing Steiner trees in
the universal cover. Assembling together these trees provides a graph, one
for each guess; the graph corresponding to the correct guess must be a
near-optimal multicut dual.

Starting with the optimal multicut dual, we iteratively modify it while
preserving its good multicut dual structure. The goal of
Sections~\ref{S:exhaustive} to~\ref{S:skeleton} is to show (1) the
existence of a ``well-behaved'' near-optimal solution and (2) how to
compute a skeleton cutting this near-optimal solution into a family of
trees.  (Section~\ref{S:homotopy-type} introduces preliminary tools, and
justifies that replacing tree-like portions of a solution with Steiner
trees in the universal cover preserves the multicut dual property.)

Actually, we compute not a single, but a number of skeleta that is a
function of $g$, $t$, and~$\varepsilon$: The main step of our algorithm,
\textsc{BuildAllSkeleta} (Section~\ref{S:skeleton}), computes a family of
skeleta of controlled length such that at least one of them intersects every cycle
of a good multicut dual.  To achieve this goal, it first enumerates the
possible topologies of the shortest multicut dual, and estimates the length
of its cycles.  For each such combination of cycle lengths and topologies,
the algorithm runs the \textsc{BuildOneSkeleton} algorithm.  If the
combination of cycle lengths and topologies corresponds to the one of a
good multicut dual, the \textsc{BuildOneSkeleton} algorithm outputs a graph
cutting each cycle of this good multicut dual.  The reader can glance at
Figures~\ref{F:skeleton1} and~\ref{F:skeleton2} to get an intuition of how
this works. This algorithm heavily relies on annular covers in both the
computations and the analysis, and ultimately on Klein's multiple source
shortest path algorithm~\cite{k-msppg-05} to do the computations in
near-linear time.

In order for this algorithm to work, the length of the cycles of the
multicut dual must be estimated; the number of possibilities must be small.
For this purpose, as a preprocessing step (Section~\ref{S:nolongcycle}), we
transform the optimal multicut dual into a near-optimal one without
\emph{long} cycle: We prove with a homology argument that each time a
near-optimal solution contains a long cycle, one can remove an edge of the
long cycle from the core and add the original cycle as a new
cycle of the multicut dual.  These cycles can then be computed separately
later, as shortest homotopic closed curves.  This is why, in the definition
of a good multicut dual, the multicut dual is split into a core
with no long cycle and a set of cycles.

Finally, the naive application of the above approach requires enumerating
all cycles of the core $C_0$ of a good multicut dual.  This
yields a correct algorithm, but with a running time doubly-exponential in
$g+t$.  To obtain a better dependence, we show that we can restrict our
attention to a much smaller family of $O(g+t)$ cycles, called an
\textit{exhaustive family}, which bears strong similarities with a
\textit{pants decomposition}.  This is described in
Section~\ref{S:exhaustive}; the proof uses topological arguments on closed
curves on surfaces that are quite different from the rest of the paper.

\section{Homotopy types}
\label{S:homotopy-type}
Recall that, starting with a shortest multicut dual~$C_\OPT$, we need to
prove that it can be transformed into another near-optimal multicut dual
satisfying additional structural properties.  In order to do this, we need
to define some operations that transform a multicut dual into another one.
Actually, it is more convenient to work on graphs that satisfy the multicut
dual property.

Let $C$ and~$D$ be two graphs drawn on~$S$ (but not necessarily embedded).  We say that $C$ and~$D$ have
the same \emphdef{homotopy type} if one can be obtained from the other via
a sequence of homotopies of the graph (vertices and edges move
continuously), edge contractions, and edge expansions (the reverse of edge
contractions---a vertex~$v$ is replaced with two new vertices $w$ and~$w'$,
connected by a new edge, and each edge incident to~$v$ is made incident to
either $w$ or~$w'$).  The following lemma is intuitively clear:
\begin{lemma}\label{L:preserve}
  Let $C$ and~$D$ be graphs drawn on~$S$ with the same homotopy type.  If
  $C$ satisfies the multicut dual property, then so does~$D$.
\end{lemma}
\begin{proof}
  The proof is easy as far as edge contractions and expansions are
  concerned, and for the homotopies it follows, for example, from the fact
  that homotopic graphs are homologous.
\end{proof}

Furthermore, in the special case of trees in the plane, we have the
following easy lemma:
\begin{lemma}\label{L:trees}
  Let $T$ and~$T'$ be two trees drawn in the plane, such that the leaves
  of~$T$ and~$T'$ are at the same positions.
  Then there is a sequence of edge contractions, edge expansions, and
  homotopies that transform $T$ into~$T'$ without moving the leaves.
\end{lemma}
\begin{proof}
  In~$T$, contract all the edges not incident with a leaf, obtaining a
  star.  Do the same for~$T'$.  Transform one star into the other with a
  homotopy. Note that this might add or remove crossings, which is fine since we are dealing with drawings.
\end{proof}

These easy lemmata provide us with the corollary below, showing that a
shortest multicut dual can be obtained as a family of Steiner trees in the
universal cover. It will not be used directly in this paper since we will
be working with approximate solutions, but it provides a good intuition as
to the idea of our algorithm.  For a graph $D$ drawn on $S$ and a set $P$
of points on the image of~$D$ but avoiding the vertices and the self-crossings of $D$, the
graph \emphdef{$\bm{D \setminus P}$} is the one obtained by cutting $D$
along all the points of $P$ on an edge of $D$.  More precisely, for each
point $p \in E(D)$ we do the following.  Let $u$ and~$v$ be two points
close to~$p$ on the edge~$e$ of~$D$ containing~$p$, so that the points $u$,
$p$, and~$v$ appear in this order.  Subdivide edge~$e$ twice by inserting
vertices $u$ and~$v$ on~$e$, and remove edge~$uv$.  Finally, push $u$
and~$v$ towards the place where~$p$ was, without changing the graph, only
its drawing on~$S$ (thus the graph has two distinct vertices that overlap).

Note that, if $T$ is a tree drawn on~$S$, it can be lifted to the universal
cover~$\tilde S$ of~$S$ as follows.  Choose an arbitrary root~$r$ of~$T$,
and a lift~$\tilde r$ of~$r$ in~$\tilde S$.  For each leaf~$\ell$ of~$T$,
lift the \emph{unique} path from~$r$ to~$\ell$ in~$T$ to~$\tilde S$
(starting at~$\tilde r$).  The result is a tree drawn on~$\tilde S$ with
the same number of vertices and edges as~$T$ (but possibly fewer
self-crossings).

\begin{corollary}\label{C:structure}
  Let $C_\OPT$ be a shortest multicut dual drawn on $S$, and $P$ a set of
  points on the edges of $C_\OPT$ (but not on self-crossing points) such
  that $C_\OPT \setminus P$ is a forest $F=(T_1,\ldots, T_m)$.  For each
  $i\in\{1,\ldots,m\}$, let $\tilde{T_i}$ be a lift of~$T_i$ in the
  universal cover of $S$, and let $\tilde{P_i}$ denote the set of lifts of
  points of~$P$ that are leaves of $\tilde{T_i}$. Then the graph consisting of the
  projections on $S$ of the Steiner trees with terminals $\tilde{P_i}$ is a
  shortest multicut dual on $S$.
\end{corollary}
\begin{proof}
  For each $i \in\{1,\ldots,m\}$, the Steiner tree with terminals
  $\tilde{P_i}$ is no longer than $\tilde{T_i}$ since it is a Steiner tree
  linking the same terminals, and it has the same homotopy type as
  $\tilde{T_i}$ by Lemma~\ref{L:trees}. Thus, the graph consisting of the
  projections of the $\tilde{T_i}$s on $S$ has the same homotopy type as
  $C_\OPT$ and is a multicut dual by Lemma~\ref{L:preserve}. By
  construction, it cannot be longer than $C_\OPT$, so it is a shortest
  multicut dual.
\end{proof}
\section{Exhaustive families}
\label{S:exhaustive}
We say that two closed curves $\gamma$ and~$\delta$ are
\emphdef{essentially crossing} on~$S$ if for each choice of closed curves
$\gamma'$ and~$\delta'$ homotopic to $\gamma$ and~$\delta$ respectively
in~$S$, the closed curves $\gamma'$ and~$\delta'$ intersect. Recall that a
\emph{cycle} in a graph drawn (or embedded) on~$S$ is a cycle in the
graph-theoretical sense (a closed walk, not reduced to a single vertex,
without vertex repetition).  A family $\Gamma$ of (abstract) cycles in a graph $D$
embedded on $S$ forms an \emphdef{exhaustive family} if, for each
cycle~$\delta$ in~$D$, either $\delta$ is a cycle of $\Gamma$ or there exists a
cycle~$\gamma$ in~$\Gamma$ such that the embeddings of $\gamma$ and~$\delta$ are essentially
crossing in~$S$.

Our algorithm to build the skeleta (Section~\ref{S:skeleton}) requires
computing exhaustive families for several embedded graphs, which for the sake of efficiency should be
small.  A trivial exhaustive family of~$D$ is the entire family of the
$(g+t)^{O(g+t)}$ cycles in~$D$.  Using this exhaustive family in
Section~\ref{S:skeleton} results in a $(1+\eps)$-approximation of
\textsc{Multicut} in time $f(\eps,g,t)\cdot n\log n$, where
$f(\eps,g,t)=\left(\log(\frac{g+t}\varepsilon)/\varepsilon\right)^{(g+t)^{O(g+t)}}$.
The purpose of this section is to give an algorithm to compute efficiently
a smaller exhaustive family, resulting in a singly-exponential~$f$ as
desired (Theorem~\ref{T:main}).  The main result of this section is the
following:

\begin{proposition}\label{P:exhaustive}
  Let $D$ be an eligible graph on~$S$ (and thus embedded). We can compute
  an exhaustive family~$\Gamma$ made of $O(g+t)$ cycles of~$D$ in time
  $(g+t)^{O(g+t)}$.  Moreover, this computation and the topology of the
  cycles of $\Gamma$ only depend on the topology of~$D$.
\end{proposition}

Some remarks: Of course, a given graph~$D$ has (in general) several
exhaustive families; the algorithm computes one such exhaustive family,
specified by a set of cycles in (the abstract graph) $D$. On the other hand, if $D$ is a cellularly embedded graph, its family of faces (or any homology basis of $D$) is \emph{not} an exhaustive family. The last part of
the proposition states that, to perform this computation, the exact
knowledge of~$D$ is unnecessary; only the topology of~$D$ is
actually used by the algorithm.  Of course, if we know the topology of~$D$
and the cycles of~$\Gamma$ (viewed as cycles in the graph~$D$), we can
infer the topology of the cycles in~$D$ (viewed as closed curves in~$S$).
Also, the assumption that $D$ is eligible is only used to
bound the number of cycles in~$\Gamma$ and the complexity of the algorithm.

The remaining part of this section is devoted to the proof of
Proposition~\ref{P:exhaustive}.  It uses arguments that are independent
from the rest of the paper.  Thus, the reader might wish to jump to the
next section in a first reading, admitting Proposition~\ref{P:exhaustive}
(as indicated above, a trivial weakened version of
Proposition~\ref{P:exhaustive}, with $(g+t)^{O(g+t)}$ cycles, still gives
an $f(\eps,g,t)\cdot n\log n$ algorithm for a
$(1+\varepsilon)$-approximation of \textsc{Multicut}, with a worse
dependence on $g$, $t$, and~$\varepsilon$ than announced in
Theorem~\ref{T:main}).

\subsection{The algorithm}

We say that two distinct cycles $\gamma$ and~$\delta$ in the embedded
graph~$D$ \emphdef{cross} at one connected component $p$ of the
intersection of the images of $\gamma$ and~$\delta$ (a path, possibly
reduced to a single vertex) if, after contracting~$p$ in both $\gamma$
and~$\delta$ into a degree-four vertex~$v$, the pieces of~$\gamma$
and~$\delta$ alternate at~$v$.  The algorithm for
Proposition~\ref{P:exhaustive} is greedy: Start with an empty set~$\Gamma$. While there exists a cycle in $D$ that does not cross any of the cycles in $\Gamma$, add it to $\Gamma$. The algorithm ends when there are no such cycles anymore.

\subsection{Some preliminary lemmas}

We say that a curve~$\delta$ is \emphdef{minimally self-crossing} if there
is no curve $\delta'$ homotopic to~$\delta$ with fewer self-crossings.
Similarly, two curves $\delta_1$ and~$\delta_2$ are \emphdef{minimally
  crossing} if there is no pair of curves $\delta'_1$ and~$\delta'_2$,
homotopic respectively to $\delta_1$ and~$\delta_2$, such that $\delta'_1$
and~$\delta'_2$ cross less than $\delta_1$ and~$\delta_2$.  We first
introduce auxiliary results. The following lemma relates monogons and
bigons to the number of crossings of closed curves.

\begin{lemma}[{Hass and Scott~\cite[Lemma~3.1]{hs-ics-85}}]\label{L:hass-scott-2}
  Let $\gamma$ and~$\delta$ be simple closed curves on a surface.  Assume
  that $\gamma$ and~$\delta$ are not minimally crossing.  Then they form a
  bigon.
\end{lemma}

The following lemma bounds the number of ``interesting'' disjoint closed curves on a surface.

\begin{lemma}[{Juvan et al.~\cite[Lemma~3.2]{jmm-scs-96}}]\label{L:pantsdecomp}
  Let $\Gamma$ be a family of pairwise disjoint simple closed curves on a
  surface with genus~$g$ and $b$~boundary components.  Assume that the
  curves in~$\Gamma$ are pairwise non-homotopic.  Then $\Gamma$ has
  $O(g+b)$ elements.
\end{lemma}

Let $N$ be a tubular neighborhood of (the image of)~$D$. The following lemma shows
that one can minimize the number of self-crossings and crossings of a
whole family of curves.

\begin{lemma}[{de Graaf and Schrijver~\cite{gs-mcmcr-97}}]\label{L:degraaf-schrijver}
Let $(\delta_1,\ldots,\delta_k)$ be a family of closed curves in~$N$.  Then
there exist closed curves $(\delta'_1,\ldots,\delta'_k)$ in general
position in~$N$ such that:
  \begin{itemize}
  \item for each~$i$, the curves $\delta_i$ and~$\delta'_i$ are homotopic;
  \item every curve $\delta'_i$ is minimally self-crossing, and
  \item every pair of curves $\delta'_i$ and~$\delta'_j$ is minimally
    crossing.
  \end{itemize}
\end{lemma}

\subsection{Proof of Proposition~\ref{P:exhaustive}}

\begin{proof}[Proof of Proposition~\ref{P:exhaustive}]
  It seems simpler to present this proof in the combinatorial setting, rather than the cross-metric one. Let
  $\Gamma$ be any maximal family of pairwise non-crossing cycles in $D$, such as the one
  computed by the greedy algorithm. Computing $\Gamma$ takes $(g+t)^{O(g+t)}$ time.  Indeed, there are
  $(g+t)^{O(g+t)}$ cycles in~$D$ (because $D$ has $O(g+t)$ vertices
  and edges, and every cycle can be represented by a permutation of a
  subset of edges), which we can enumerate in $(g+t)^{O(g+t)}$ time.  The
  algorithm can easily determine whether two cycles in~$D$ cross in
  $O(g+t)$ time.  It is clear that this algorithm only needs the topology
  of~$D$, and that the topology of the resulting cycles also only depends
  on the topology of $D$.

  We now prove that $\Gamma$ is made of $O(g+t)$ cycles.  No pair of cycles
  in~$\Gamma$ is essentially crossing in~$S$, which implies that they can
  all be made simple and disjoint by homotopies on the tubular neighborhood~$N$, by
  Lemma~\ref{L:degraaf-schrijver}.  These considerations lead to a
  family~$\Gamma'$ of pairwise disjoint simple closed curves in~$N$.

  Moreover, no two distinct cycles in~$D$ are homotopic in~$D$ (this is
  a basic fact valid in any graph),  so
  the cycles in~$\Gamma$ are pairwise non-homotopic in~$D$, and thus the
  cycles in~$\Gamma'$ are pairwise non-homotopic in~$N$.  Since $N$ has
  genus at most~$g$, and at most $t$ boundary components because every face
  of~$D$ contains a terminal (because $D$ is eligible), Lemma~\ref{L:pantsdecomp} implies that $\Gamma$ has
  $O(g+t)$ cycles.

  There remains to prove that $\Gamma$ is an exhaustive family.  Let
  $\delta$ be a cycle in~$D$ but not in~$\Gamma$.  By
  construction, there exists $\gamma\in\Gamma$ such that $\gamma$
  and~$\delta$ cross. To conclude the proof, it suffices to proves that $\gamma$ and $\delta$ are essentially crossing in $S$.

  \begin{figure}
    \begin{center}
      \def\svgwidth{10cm}
      \begingroup%
  \makeatletter%
  \providecommand\color[2][]{%
    \errmessage{(Inkscape) Color is used for the text in Inkscape, but the package 'color.sty' is not loaded}%
    \renewcommand\color[2][]{}%
  }%
  \providecommand\transparent[1]{%
    \errmessage{(Inkscape) Transparency is used (non-zero) for the text in Inkscape, but the package 'transparent.sty' is not loaded}%
    \renewcommand\transparent[1]{}%
  }%
  \providecommand\rotatebox[2]{#2}%
  \ifx\svgwidth\undefined%
    \setlength{\unitlength}{103.30992432bp}%
    \ifx\svgscale\undefined%
      \relax%
    \else%
      \setlength{\unitlength}{\unitlength * \real{\svgscale}}%
    \fi%
  \else%
    \setlength{\unitlength}{\svgwidth}%
  \fi%
  \global\let\svgwidth\undefined%
  \global\let\svgscale\undefined%
  \makeatother%
  \begin{picture}(1,0.56321743)%
    \put(0,0){\includegraphics[width=\unitlength]{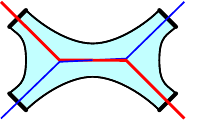}}%
    \put(0.21063531,0.36111984){\color[rgb]{0,0,0}\makebox(0,0)[lb]{\smash{$\gamma$}}}%
    \put(0.21478365,0.18688682){\color[rgb]{0,0,0}\makebox(0,0)[lb]{\smash{$\delta$}}}%
    \put(0.0682067,0.52567325){\color[rgb]{0,0,0}\makebox(0,0)[lb]{\smash{$a^\gamma_1$}}}%
    \put(0.74895467,0.53968387){\color[rgb]{0,0,0}\makebox(0,0)[lb]{\smash{$a^\delta_1$}}}%
    \put(0.75448571,0.00730508){\color[rgb]{0,0,0}\makebox(0,0)[lb]{\smash{$a^\gamma_2$}}}%
    \put(0.06585061,0.01421929){\color[rgb]{0,0,0}\makebox(0,0)[lb]{\smash{$a^\delta_2$}}}%
    \put(0.42082114,0.30165941){\color[rgb]{0,0,0}\makebox(0,0)[lb]{\smash{$p$}}}%
  \end{picture}%
\endgroup%

      \caption{A picture of the situation in~$P$, in the proof of
        Proposition~\ref{P:exhaustive}.}
        \label{F:proof-exhausting}
    \end{center}
\end{figure}
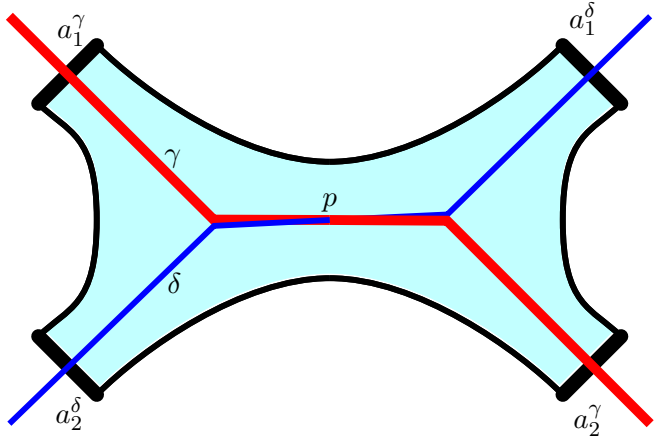

  Let us first prove that the cycles $\gamma$ and~$\delta$
  are essentially crossing in the tubular neighborhood~$N$. Since $N$ is a regular neighborhood of
  $D$, this is intuitively obvious, but the proof is somewhat
  heavy-handed.  Let $p$ be a connected component of the
  intersection of the images of $\gamma$ and~$\delta$ where these
  cycles cross. Let $P$ be a disk neighborhood of~$p$ in~$N$ as shown
  in Figure~\ref{F:proof-exhausting}, so that the boundary of~$P$ is
  made of pieces of boundaries of~$N$ and of arcs in~$N$, four of
  which intersect $\gamma$ or~$\delta$, which we denote by
  $a^1_\gamma$, $a^1_\delta$, $a^2_\gamma$, and $a^2_\delta$ (the
  subscript denoting which cycle crosses the arc); since $\gamma$
  and~$\delta$ cross, without loss of generality the arcs appear in
  that order when walking along the boundary of~$P$.  Let $\tilde P$
  be a lift of~$P$ in the universal cover~$\tilde N$ of~$N$, let
  $\tilde a^1_\gamma$, $\tilde a^1_\delta$, $\tilde a^2_\gamma$, and
  $\tilde a^2_\delta$ the corresponding lifts of the arcs, and let
  $\tilde\gamma$ and~$\tilde\delta$ be lifts of $\gamma$ and~$\delta$
  entering~$\tilde P$.  Each of these lifts of the arcs separates the
  universal cover of~$\bar D$ into two connected components, and is
  crossed exactly once by exactly one of $\tilde\gamma$
  and~$\tilde\delta$.  Any homotopy in~$N$ between~$\gamma$ and
  another closed curve~$\gamma'$ lifts to a homotopy between
  $\tilde\gamma$ and another lift $\tilde\gamma'$, which still has to
  cross $\tilde a^1_\gamma$ and~$\tilde a^2_\gamma$, and similarly if
  $\delta'$ is homotopic to~$\delta$ in~$N$ we obtain a lift
  $\tilde\delta'$ that crosses $\tilde a^1_\delta$ and~$\tilde
  a^2_\delta$.  This implies that $\tilde\gamma'$ and~$\tilde\delta'$
  cross, and thus $\gamma$ and~$\delta$ are essentially crossing
  in~$N$.

  Now, let $\gamma''$ and~$\delta''$ be curves in~$N$, homotopic to
  $\gamma$ and~$\delta$ in~$N$, and minimally crossing given these
  constraints.  By the preceding paragraph, $\gamma''$ and~$\delta''$
  have to cross, and we can choose $\gamma''$ and~$\delta''$ to be
  simple, again by Lemma~\ref{L:degraaf-schrijver}.  They form no
  bigon in~$N$.  They also cannot form a bigon in~$S$: Indeed, no
  such bigon could be entirely in~$N$, so any bigon has to contain at
  least one connected component of $S\setminus N$.
  Each such component
  contains at least one boundary component of~$S$, because $D$ is eligible;
  this contradicts the
  definition of a bigon.  This implies, by Lemma~\ref{L:hass-scott-2},
  that $\gamma''$ and~$\delta''$ are minimally crossing in~$S$.  In
  other words, $\gamma$ and~$\delta$ are essentially crossing in~$S$,
  as desired.
\end{proof}

\section{Near-optimal solution with no long cycle}
\label{S:nolongcycle}
We say that a cycle~$\gamma$ of a graph $D$ drawn on $(S,G)$ is
\emphdef{long} if $|\gamma^o|=O(|\gamma|\varepsilon/(g+t))$ (for some
universal, well-chosen constant hidden in the $O(\cdot)$ notation).  As a
preprocessing step for the computation of the skeleta
(Section~\ref{S:skeleton}), we need to make sure that there is a good
multicut dual whose core has no long cycle:

\begin{proposition}\label{P:shortcycles}
  There exists a good multicut dual $C$ such that
  $C_0$ is eligible and has no long cycle.
\end{proposition}

\begin{proof}
  Initially, let $C:=\{C_0\}:=\{C_\OPT\}$.  (Note that $C_\OPT$ is not
  necessarily connected.)  We have that $C$ is a good multicut dual that is
  eligible.  Whenever $C_0$ contains a long cycle $\gamma$, we do the
  following: We remove one of the edges of~$\gamma$ from~$C_0$, iteratively
  prune vertices of degree one (by removing such vertices and their
  incident edges), and finally add to~$C$ (as a separate component) the
  cycle~$\gamma^o$ (as a graph, it is a loop; the image of the vertex is
  located arbitrarily on the image of~$\gamma^o$). We refer to
  Figure~\ref{F:nolongcycle} for a picture in the planar case, the
  surface-embedded case bearing no difference.

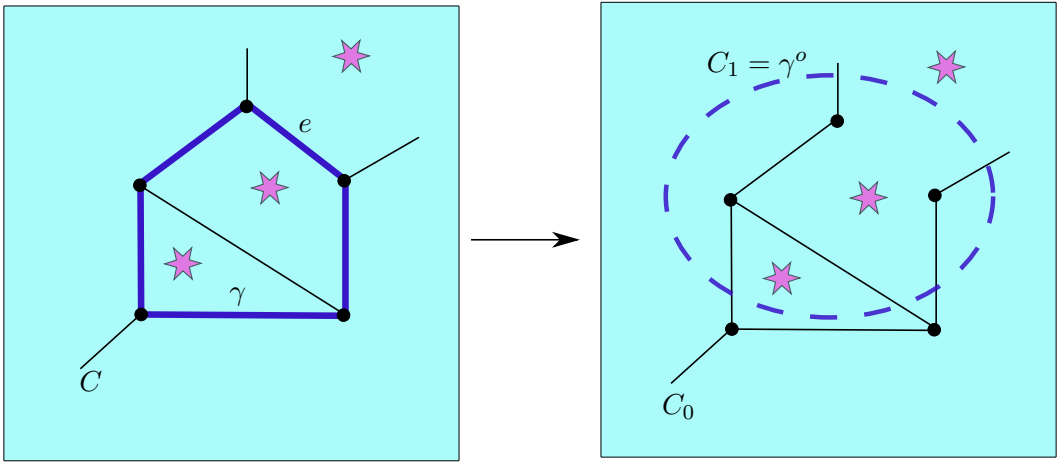
\begin{figure}
\centering
\def\svgwidth{14cm}
\begingroup%
  \makeatletter%
  \providecommand\color[2][]{%
    \errmessage{(Inkscape) Color is used for the text in Inkscape, but the package 'color.sty' is not loaded}%
    \renewcommand\color[2][]{}%
  }%
  \providecommand\transparent[1]{%
    \errmessage{(Inkscape) Transparency is used (non-zero) for the text in Inkscape, but the package 'transparent.sty' is not loaded}%
    \renewcommand\transparent[1]{}%
  }%
  \providecommand\rotatebox[2]{#2}%
  \newcommand*\fsize{\dimexpr\f@size pt\relax}%
  \newcommand*\lineheight[1]{\fontsize{\fsize}{#1\fsize}\selectfont}%
  \ifx\svgwidth\undefined%
    \setlength{\unitlength}{3490.11774641bp}%
    \ifx\svgscale\undefined%
      \relax%
    \else%
      \setlength{\unitlength}{\unitlength * \real{\svgscale}}%
    \fi%
  \else%
    \setlength{\unitlength}{\svgwidth}%
  \fi%
  \global\let\svgwidth\undefined%
  \global\let\svgscale\undefined%
  \makeatother%
  \begin{picture}(1,0.43647546)%
    \lineheight{1}%
    \setlength\tabcolsep{0pt}%
    \put(0,0){\includegraphics[width=\unitlength,page=1]{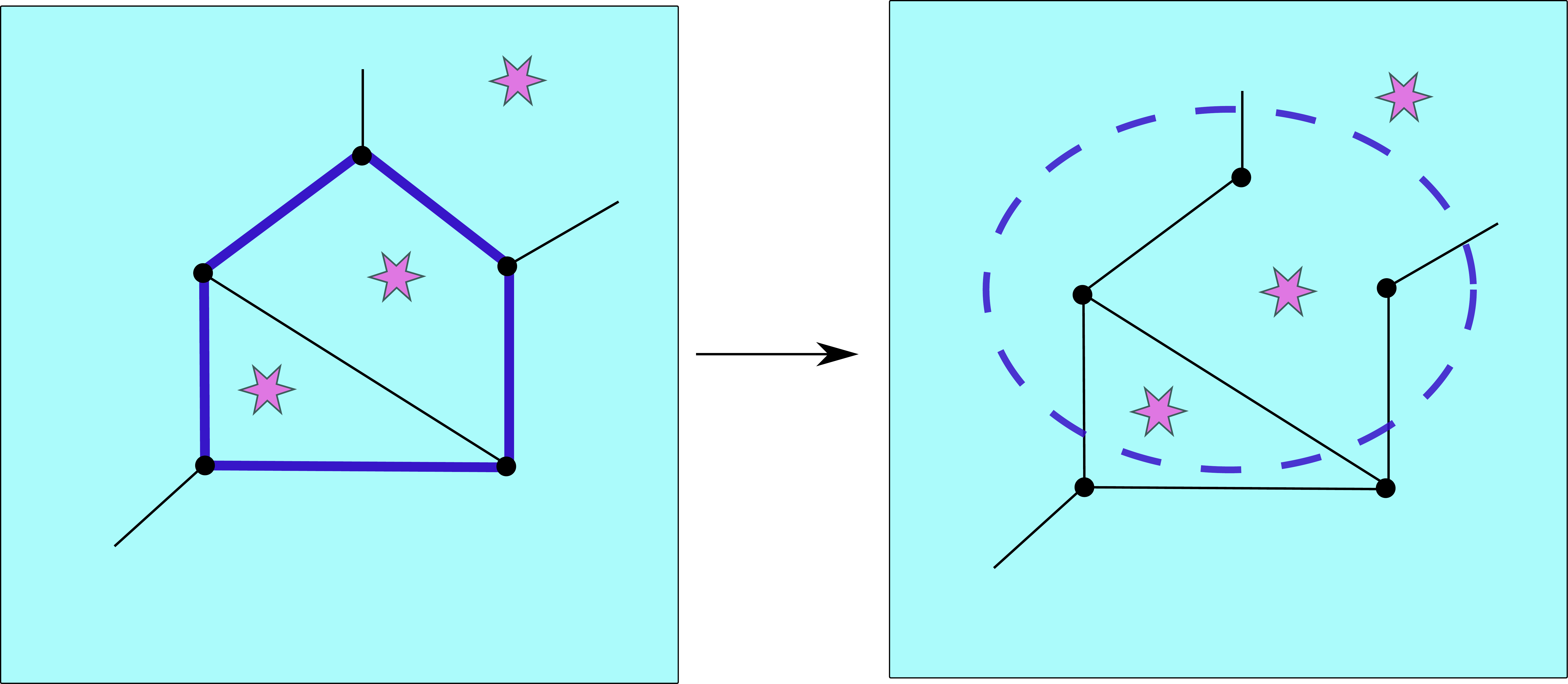}}%
    \put(0.07236391,0.06628582){\color[rgb]{0,0,0}\makebox(0,0)[lt]{\lineheight{0}\smash{\begin{tabular}[t]{l}$C$\end{tabular}}}}%
    \put(0.22259031,0.15311548){\color[rgb]{0,0,0}\makebox(0,0)[t]{\lineheight{0}\smash{\begin{tabular}[t]{c}$\gamma$\end{tabular}}}}%
    \put(0.28642013,0.31288208){\color[rgb]{0,0,0}\makebox(0,0)[t]{\lineheight{0}\smash{\begin{tabular}[t]{c}$e$\end{tabular}}}}%
    \put(0.62546887,0.0445784){\color[rgb]{0,0,0}\makebox(0,0)[lt]{\lineheight{0}\smash{\begin{tabular}[t]{l}$C_0$\end{tabular}}}}%
    \put(0.71578283,0.37018965){\color[rgb]{0,0,0}\makebox(0,0)[t]{\lineheight{0}\smash{\begin{tabular}[t]{c}$C_1=\gamma^o$\end{tabular}}}}%
  \end{picture}%
\endgroup%

\caption{The edge $e\subseteq \gamma$ is removed
  from $C_0$, and the cycle $\gamma^o$ is added to~$C$; the cycle
  $\gamma^o$ may cross $C_0$.}
  \label{F:nolongcycle}
\end{figure}

Each iteration removes at least one edge in~$C_0$, which initially is equal
to the good multicut dual~$C_\OPT$ and therefore has $O(g+t)$ edges.
Hence, there are $O(g+t)$ iterations.  The new graph consists of the
remaining graph~$C_0$ and of the set of cycles that were added at each
step, $C_1,\ldots,C_{k}$.  By construction, $C_0$ has no long cycle.  Let
us argue that this final graph $C=\{C_0,\ldots,C_k\}$ is a good multicut
dual, and that $C_0$ is eligible.

Let us prove Condition~\ref{c:small} (each of $C_0,C_1,\ldots,C_k$ is
small).  Conditions \ref{c:small}\ref{sm:genpos}
and~\ref{c:small}\ref{sm:nodeg01} are satisfied by construction.
Condition~\ref{c:small}\ref{sm:compl} is satisfied for~$C_0$ because it is
satisfied for $C_\OPT$, and because $C_0$ is a subgraph of~$C_\OPT$.  It is
satisfied for~$C_1,\ldots,C_k$ by construction.  Let us prove that
Condition~\ref{c:small}\ref{sm:cross} holds.  As above, it is satisfied
for~$C_0$ because it is satisfied for $C_\OPT$, and because $C_0$ is a
subgraph of~$C_\OPT$; let us prove it for the cycles $C_1,\ldots,C_k$.
Recall that $\gamma^o$, being a shortest homotopic closed curve, crosses
arcs of~$K$ minimally, and therefore the number of crossings between
$\gamma^o$ and~$K$ does not exceed the number of crossings between $\gamma$
and~$K$. In particular, the number of intersections of each cycle $C_i$
with each arc of $K$ is at most the number of intersections of $C_{\OPT}$
with each arc of $K$, that is, $O(g+t)$.  So Condition~\ref{c:small} holds.

Condition~\ref{c:cycles} is satisfied by construction.

Condition~\ref{c:nbcycles} holds because the number of iterations is
$O(g+t)$ and each iteration adds one cycle.

Let us prove that Condition~\ref{c:length} holds.  At each iteration, the
length of $C$ increases by at most the length of $\gamma^o$, which
satisfies
\[|\gamma^o|=O(|\gamma|\varepsilon/(g+t)) =O(\OPT\cdot\varepsilon/(g+t)). \] %
This implies that after the $O(g+t)$ iterations, the length of the graph
has increased by at most $O(\varepsilon\cdot\OPT)$, as desired.

Let us prove that Condition~\ref{c:homology} is satisfied.  Let
$C'_\OPT$ be an even subgraph of $C_\OPT$.  We need to prove that $C'_\OPT$
has the same homology class as some even subgraph of~$C$.  Initially,
$C=\{C_\OPT\}$, so this property was trivially true at the beginning.  It
thus suffices to prove that this property is maintained when performing an
iteration that removes one edge~$e$ that belongs to a cycle~$\gamma$ and
adds $\gamma^o$ to~$C$.  Before this iteration, we know that $C'_\OPT$ has
the same homology class as some subgraph $C'$ of~$C$.  If $C'$ does not
contain the chosen edge~$e$, this is trivially true after the replacement.
Otherwise, $C'=((C'\setminus e)\Delta(\gamma\setminus e))+\gamma$ (where
$\Delta$ denotes the symmetric difference),
which is homologous to $((C'\setminus e)\Delta(\gamma\setminus e))+\gamma^o$, an
even subgraph of the new~$C$.

Finally, we have proved that $C_0$ is small.  Moreover, being a subgraph
of~$C_\OPT$, it is embedded, and each face of~$C_0$ contains a terminal.
Thus, $C_0$ is eligible.
\end{proof}

\section{Building skeleta}
\label{S:skeleton}
In this section, we describe the \textsc{BuildAllSkeleta} algorithm whose
properties are as follows.
\begin{proposition}\label{P:skeleton}
  In $\paren{\frac{g+t}\varepsilon}^{O(g+t)}\cdot n\log n$ time, the 
  \textsc{BuildAllSkeleta} algorithm
  builds a family of $\paren{\frac{g+t}\varepsilon}^{O(g+t)}$ graphs, called
  \emph{skeleta}, drawn on~$S$, each with $O(g+t)$ vertices and edges and $O((g+t)^3n)$ complexity,
  such that at least one of them, denoted by $\Sk(C_0)$, satisfies the
  following conditions:
  \begin{itemize}
  \item $\Sk(C_0)$ has length $O((g+t)\OPT)$;
  \item There exists a good multicut dual~$C'=\{C'_0,C'_1,\ldots,C'_k\}$
    (with $C'_0$ not necessarily eligible) such that $\Sk(C_0)$ intersects every 
    cycle of~$C'_0$.
  \end{itemize}
\end{proposition}

In this proposition and the next one, the data structure used to represent
a skeleton and its drawing is its abstract graph, and for each of its edges
$e$, the ordered sequence of the edges of $G$ crossed by $e$ in the
drawing, together with the orientation of each crossing. The bound on the
complexity of the skeleta ensures that this encoding has a controlled size.

This data structure does not encode the precise location of the skeleton
inside each face of~$G$ but it is sufficient for the purpose of the
algorithm. More precisely, for the proofs in this section, we rely on the
precise location of the skeleta, but in the actual algorithms (see
Section~\ref{S:portals}), this data structure suffices.\footnote{While such
  an ambiguity in the encoding is customary when handling embedded graphs
  in a cross-metric setting, it is somewhat worsened here by the fact that
  we handle drawings and not embeddings. For example, edges might cross
  themselves arbitrarily many times within a face of $G$, and the encoding
  will not reflect it.}

The proof of Proposition~\ref{P:skeleton} relies on the exhaustive
families introduced in Section~\ref{S:exhaustive}; the proof also
extensively uses covering spaces and algorithms to compute shortest
homotopic curves, described in Appendix~\ref{A:covering}.


\subsection{Definition of a skeleton}\label{SS:skeleta}


Let $D$ be an eligible graph embedded on $S$ that has no long cycle. A \emphdef{skeleton} $\Sk(D)$ is a graph drawn on $S$ that satisfies the following properties:

\begin{itemize}
  \item $\Sk(D)$ has $O(g+t)$ vertices and edges;
  \item $\Sk(D)$ has complexity $O((g+t)^3n)$.
  \item $\Sk(D)$ has length $O((g+t)|D|)$;
  \item there exists a graph~$D'$ drawn on~$S$ such that:
    \begin{itemize}
    \item $D'$ has the same homotopy type as~$D$,
    \item $|D'|\leq(1+\varepsilon)|D|$,
    \item $\Sk(D)$ intersects every cycle of~$D'$,
    \item $D'$ is small.
    \end{itemize}
\end{itemize}

We now introduce the algorithm we use, called \textsc{BuildOneSkeleton},
for computing a single skeleton,
and will prove the following key proposition.

\begin{proposition}\label{P:oneskeleton}
  Let $D$ be an eligible graph embedded on $S$ that
  has no long cycle.  In $(g+t)^{O(g+t)}n\log n$ time, the
  \textsc{BuildOneSkeleton} algorithm computes a skeleton $\Sk(D)$.

  Moreover, the construction of $\Sk(D)$ does not depend on~$D$, but only on one out of
$(\frac{g+t}\varepsilon)^{O(g+t)}$ possible choices of parameters.
\end{proposition}

Note that we can apply Proposition~\ref{P:oneskeleton} with $D=C_0$.
Hence, we define \textsc{BuildAllSkeleta} to be the algorithm that, for
every possible choice of parameters, computes the corresponding skeleton
using the \textsc{BuildOneSkeleton} algorithm and returns the family
consisting of all these skeleta.  Assuming this
Proposition~\ref{P:oneskeleton}, we can directly prove
Proposition~\ref{P:skeleton}:

\begin{proof}[Proof of Proposition~\ref{P:skeleton}, assuming Proposition~\ref{P:oneskeleton}]
  The proof is immediate and only involves definition-chasing: we apply
  Proposition~\ref{P:oneskeleton} for all the possible choices of
  parameters, which can be done in
  $\paren{\frac{g+t}\varepsilon}^{O(g+t)}\cdot n\log n$ time. All the
  resulting skeleta satisfy the bounds on the number of vertices and edges
  and the complexity. Out of these, for the parameters corresponding to the
  $C_0$ of Proposition~\ref{P:shortcycles}, we obtain $\Sk(C_0)$, which has
  length $O((g+t)|C_0|)=O((g+t)OPT)$, and we denote by $C'_0$ the graph
  $D'$ associated to the skeleton~$\Sk(C_0)$.  By construction, $\Sk(C_0)$
  intersects every cycle of~$C'_0$.  There remains to prove that
  $C'=(C'_0,C_1, \ldots, C_k)$ is a good multicut dual.

  Conditions~\ref{c:small}, \ref{c:cycles},
  \ref{c:nbcycles}, and~\ref{c:length} follow from the construction and
  from the properties of~$C'_0$.  Because $C'_0$ has the same homotopy type
  as $C_0$, Lemma~\ref{L:preserve} proves that Condition~\ref{c:homology}
  is fulfilled.
\end{proof}

The algorithm \textsc{BuildOneSkeleton} takes as input a set of
\emphdef{parameters}, which we now describe. The first of these is the
topology of $D$. By Proposition~\ref{P:exhaustive}, an exhaustive
family~$\Gamma$ can be computed for~$D$ just with the knowledge of its
topology, and it consists of $O(g+t)$ cycles.  There is one additional
parameter per two-sided cycle in~$\Gamma$, which corresponds to an estimate
of its length up to a precision $(1+\varepsilon)$.  More precisely, fix a
two-sided cycle $\gamma \in \Gamma$.  Since $D$ has no long cycle, we have
$|\gamma|= O(|\gamma^o|(g+t)/\varepsilon)$, hence $|\gamma|$ lies within
one of the following \emphdef{ranges}:

\[[|\gamma^o|, (1+\varepsilon) |\gamma^o|),\quad
[(1+\varepsilon)|\gamma^o|, (1+\varepsilon)^2 |\gamma^o|),\quad \ldots,
\quad [(1+\varepsilon)^\ell|\gamma^o|, (1+\varepsilon)^{\ell+1}
|\gamma^o|),\]
for some non-negative integer $\ell$ satisfying
$(1+\varepsilon)^{\ell+1}=\Theta((g+t)/\varepsilon)$, and thus
$\ell= \Theta(\log\left(\frac{g+t}\varepsilon\right)/\varepsilon)$.  Therefore,
for a two-sided cycle $\gamma$, we associate an integer parameter
$\theta(\gamma):=\lfloor \log_{1+\varepsilon}(|\gamma_i|/|\gamma_i^o|)
\rfloor$, between 0 and~$\ell$, indicating in which range $|\gamma|$ lies.

We now describe the \textsc{BuildOneSkeleton} algorithm for constructing
$\Sk(D)$.  First, we compute the exhaustive family~$\Gamma$ associated
to~$D$.  We now look at each cycle~$\gamma\in\Gamma$ in turn, and add some
vertices and edges in~$\Sk(D)$, as follows.

If $\gamma$ is one-sided, we just add it to~$\Sk(D)$.

Otherwise, $\gamma$ is two-sided.  It is also non-contractible, because
otherwise a face of~$D$ would be a disk, contradicting the eligibility
of~$D$.  Let $\hat{S}_\gamma$ be the annular cover of~$S$ corresponding
to~$\gamma$
; let $\hat\gamma$ be a cyclic lift of $\gamma$ in this covering space. Let
$\hat p$ be an arbitrary lift of an arc of~$K$ that ``connects the two
boundaries of the annular cover'', or more  formally that crosses a shortest
closed curve homotopic to~$\hat\gamma$; such a lift exists because
otherwise $\hat\gamma$ would be contractible, and furthermore, every closed
curve homotopic to~$\hat\gamma$ crosses it (see Appendix~\ref{A:covering}).
Moreover, since $\hat{p}$ is a lift of an arc, it
intersects only a finite number of faces of $\hat{S}_{\gamma}$. For every
face~$f$ of~$\hat{S}_\gamma$ crossed by~$\hat p$, let $\hat\gamma_f$ be a
shortest non-contractible closed curve in the annulus~$\hat S_\gamma$ that
goes through~$f$.  Among all the closed curves $\hat\gamma_f$ of length at
most $(1 + \varepsilon)^{\theta(\gamma) +1} |\gamma|^o$, let $\hat\gamma^L$
and $\hat\gamma^R$ be ones corresponding to faces $f$ which are the
leftmost and rightmost along $\hat p$ (see
Figure~\ref{F:skeleton1}). Without loss of generality, these closed curves
can be chosen to be simple and disjoint. Finally, let $\hat A_\gamma$ be
the annulus bounded by $\hat\gamma^L$ and $\hat\gamma^R$, and let~$\hat D$
be the lift of~$D$ containing~$\hat\gamma$.  We add the projections
$\gamma^L$ and $\gamma^R$ of the curves $\hat\gamma^L$ and~$\hat\gamma^R$
to~$\Sk(D)$ (see Figure~\ref{F:skeleton1}, left). Moreover, if
on each side of the two-sided cycle~$\hat\gamma$ there is at least one edge
of~$\hat D$ incident to~$\hat\gamma$, we add the projection of the shortest
path in~$\hat A_\gamma$ between $\hat\gamma^L$ and $\hat\gamma^R$ to
$\Sk(D)$ (see Figure~\ref{F:skeleton1}, right).

Note that this construction only depends on the topology of $D$ and the
parameters $\theta(\gamma)$: Indeed, the topology of $D$ suffices to
compute the exhaustive family~$\Gamma$ and the topology of its
cycles. Furthermore, it specifies the length of $\gamma^o$ and whether a
cycle $\gamma$ of $\Gamma$ is one-sided or two-sided. If $\gamma$ is
two-sided, the knowledge of its topology also allows us to compute its
annular cover $\hat S_{\gamma}$.  Then, to define and compute the cycles
$\hat\gamma^L$ and $\hat\gamma^R$, one only needs to additionally know the
range of $\gamma$.  Whether a shortest path between $\hat\gamma^L$ and
$\hat\gamma^R$ is added to~$\Sk(D)$ can be deduced just by knowing the
topology of~$D$. Indeed, this is the case if and only if, in~$D$, there are
edges attached to vertices of~$\gamma$ that leave~$\gamma$ on both of its
sides: The ``only if'' part is trivial; for the ``if'' part, this follows
from the fact that any edge leaving~$\hat\gamma$ on its left (resp.,
right) must cross~$\hat\gamma^L$ (resp., $\hat\gamma^R$), because
$D$ has no degree-one vertex and no face that is a disk (by eligibility).

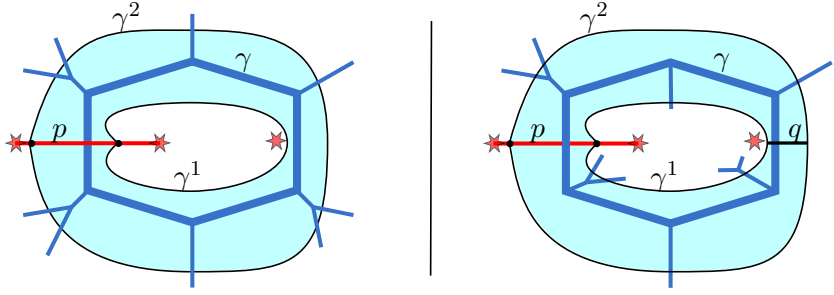
\begin{figure}
  \begin{center}
    \def\svgwidth{11cm}
    \begingroup%
  \makeatletter%
  \providecommand\color[2][]{%
    \errmessage{(Inkscape) Color is used for the text in Inkscape, but the package 'color.sty' is not loaded}%
    \renewcommand\color[2][]{}%
  }%
  \providecommand\transparent[1]{%
    \errmessage{(Inkscape) Transparency is used (non-zero) for the text in Inkscape, but the package 'transparent.sty' is not loaded}%
    \renewcommand\transparent[1]{}%
  }%
  \providecommand\rotatebox[2]{#2}%
  \newcommand*\fsize{\dimexpr\f@size pt\relax}%
  \newcommand*\lineheight[1]{\fontsize{\fsize}{#1\fsize}\selectfont}%
  \ifx\svgwidth\undefined%
    \setlength{\unitlength}{385.21790288bp}%
    \ifx\svgscale\undefined%
      \relax%
    \else%
      \setlength{\unitlength}{\unitlength * \real{\svgscale}}%
    \fi%
  \else%
    \setlength{\unitlength}{\svgwidth}%
  \fi%
  \global\let\svgwidth\undefined%
  \global\let\svgscale\undefined%
  \makeatother%
  \begin{picture}(1,0.33129976)%
    \lineheight{1}%
    \setlength\tabcolsep{0pt}%
    \put(0,0){\includegraphics[width=\unitlength,page=1]{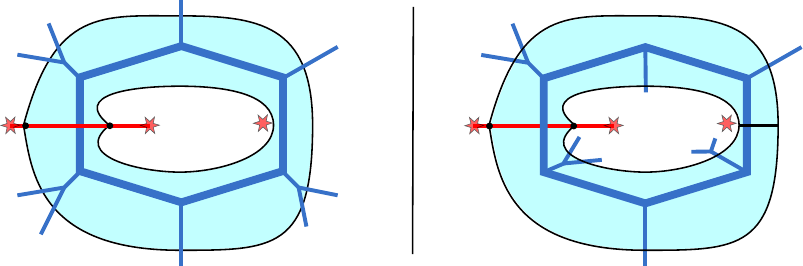}}%
    \put(0.20357214,0.12381414){\color[rgb]{0,0,0}\makebox(0,0)[lt]{\lineheight{0}\smash{\begin{tabular}[t]{l}$\gamma^L$\end{tabular}}}}%
    \put(0.27781571,0.26836781){\color[rgb]{0,0,0}\makebox(0,0)[lt]{\lineheight{0}\smash{\begin{tabular}[t]{l}$\gamma$\end{tabular}}}}%
    \put(0.12932858,0.31655238){\color[rgb]{0,0,0}\makebox(0,0)[lt]{\lineheight{0}\smash{\begin{tabular}[t]{l}$\gamma^R$\end{tabular}}}}%
    \put(0.05656008,0.18232398){\color[rgb]{0,0,0}\makebox(0,0)[lt]{\lineheight{0}\smash{\begin{tabular}[t]{l}$p$\end{tabular}}}}%
    \put(0.77962983,0.12349588){\color[rgb]{0,0,0}\makebox(0,0)[lt]{\lineheight{0}\smash{\begin{tabular}[t]{l}$\gamma^L$\end{tabular}}}}%
    \put(0.69456931,0.31574244){\color[rgb]{0,0,0}\makebox(0,0)[lt]{\lineheight{0}\smash{\begin{tabular}[t]{l}$\gamma^R$\end{tabular}}}}%
    \put(0.85584008,0.26804955){\color[rgb]{0,0,0}\makebox(0,0)[lt]{\lineheight{0}\smash{\begin{tabular}[t]{l}$\gamma$\end{tabular}}}}%
    \put(0.63458444,0.18200572){\color[rgb]{0,0,0}\makebox(0,0)[lt]{\lineheight{0}\smash{\begin{tabular}[t]{l}$p$\end{tabular}}}}%
    \put(0.94623381,0.18200572){\color[rgb]{0,0,0}\makebox(0,0)[lt]{\lineheight{0}\smash{\begin{tabular}[t]{l}$q$\end{tabular}}}}%
  \end{picture}%
\endgroup%

    \caption{The construction of the skeleton~$\Sk(D)$, as represented on the
      surface~$S$.  The cycle~$\gamma$ is in thick lines.  The closed
      curves $\gamma^L$ and~$\gamma^R$ are chosen so that, among the closed
      curves with the same homotopy class as~$\gamma$ and within the same
      range or lower, they cross~$p$ as left or as right as possible,
      respectively.  Left: In this case, $\gamma$ has incident edges only
      on one side, so only $\gamma^L$ and~$\gamma^R$ are added to~$\Sk(D)$.  The same construction would be used if $\gamma$ were
      incident to no other edge of~$C_0$. Right: In this case, $\gamma$ has incident
      edges on both of its sides, and thus not only $\gamma^L$
      and~$\gamma^R$ are added to~$\Sk(D)$, but also a shortest path
      (denoted by~$q$ here) connecting~$\gamma^L$ and~$\gamma^R$ in the
      region bounded by $\gamma^L$ and~$\gamma^R$.  }
       \label{F:skeleton1}
  \end{center}
\end{figure}

\subsection{Bounding the number of skeleta}

In line with the description of the \textsc{BuildOneSkeleton} algorithm, an
eligible topology and a choice of a range for every two-sided cycle of an
exhaustive family associated to that topology is called a \emphdef{choice
  of parameters}. Out of all the possible outputs of
\textsc{BuildOneSkeleton} computed with all the possible choices of
parameters, one will have the choice of parameters corresponding to $D$,
i.e., the correct topology and the correct choice of range for each cycle
of $D$.  We let $\Sk(D)$ be the output of \textsc{BuildOneSkeleton} for
this choice of parameters.

\begin{lemma}\label{L:numskeleta}
The number of possible choices of parameters is $\paren{\frac{g+t}\varepsilon}^{O(g+t)}$.
\end{lemma}
\begin{proof}
  By Lemma~\ref{L:eligibility-struct}, there are $(g+t)^{O(g+t)}$ eligible topologies,
  and to a given topology corresponds an exhaustive family $\Gamma$.
  Moreover, the number of ranges for every cycle is
  $\ell=O(\log\left(\frac{g+t}\varepsilon\right)/\varepsilon)$, and the
  number of cycles in~$\Gamma$ is $O(g+t)$.  As noted above, whether we
  need to add a shortest path between $\hat\gamma^L$ and~$\hat\gamma^R$ can
  be deduced from~$\gamma$ and from the topology of~$D$.  Thus, the total
  number of choices of parameters is $\paren{\frac{g+t}\varepsilon}^{O(g+t)}$, as
  desired.
\end{proof}

\subsection{Complexity}

\begin{lemma}\label{L:build-skeleton}
  The \textsc{BuildOneSkeleton} algorithm runs
  in time $(g+t)^{O(g+t)} \cdot n \log n$ and its output has complexity $O((g+t)^3n)$.
\end{lemma}
\begin{proof}
  We first compute the exhaustive family~$\Gamma$ corresponding to the
  graph~$D$, using Proposition~\ref{P:exhaustive}, in $(g+t)^{O(g+t)}$
  time.  It is made of $O(g+t)$ cycles, each of them crossing~$K$ a number
  of times that is $O((g+t)^2)$.  (For the purpose of this computation, we
  only need
 to keep the information of the sequences of the arcs of~$K$ crossed by the
 cycles in~$\Gamma$, together with
 the orientation of the crossings.)

 For each such cycle~$\gamma$, we
  compute the curve $\hat\gamma$, as well as $\hat\gamma^L$ and
  $\hat\gamma^R$, and possibly the shortest path inbetween, as follows (see
  Appendix~\ref{S:annular}).  We first compute the shortest homotopic closed curve~$\gamma^o$ in
  $O((g+t)jn\log n)$ time, where $j$ is the number of crossings between
  $\gamma$ and~$K$, as explained in Lemmas \ref{L:annulus}
  and~\ref{L:mobius}. By eligibility, we
  have $j=O((g+t)^2)$.
  The curves $\hat\gamma_f$ remain in the relevant region used to
  compute~$\gamma^o$, so, using the techniques of Lemmas \ref{L:annulus}
  and~\ref{L:mobius}  (which rely on Klein's multiple shortest path
  algorithm~\cite{k-msppg-05}), we can also compute, in the same amount of
  time, the length of all the closed curves $\hat\gamma_f$, and thus also
  the closed curves $\hat\gamma^L$ and~$\hat\gamma^R$.  We can actually
compute the curves $\hat\gamma^L$ and~$\hat\gamma^R$ in the cross-metric
  surface corresponding to the relevant region, and force them to be simple
  (by a shortest path computation) and disjoint (by computing one of them
  first, and forbidding the second one to cross the first one).
  Finally, if needed, we can compute a shortest path between $\hat\gamma^L$
  and~$\hat\gamma^R$.   We
  then project these curves to the surface~$S$, keeping only, for each edge~$e$
  of~$\Sk(D)$, the ordered sequence of oriented edges of~$G$ crossed
  by~$e$.  (The projections of the curves $\hat\gamma^L$, $\hat\gamma^R$,
  and possibly the shortest path, may have crossings.)

  Since $\Sk(D)$ is made of $O(g+t)$ projections of shortest paths in a relevant region (see Appendix~\ref{A:covering}) of size $O((g+t)^2n)$, the bound on the complexity follows.
\end{proof}

\subsection{Properties of~$\Sk(D)$}

For an eligible graph $D$ with no long cycle, recall that $\Sk(D)$ is  computed with the choice of parameters corresponding to $D$. To prove Proposition~\ref{P:skeleton}, it remains to prove that
$\Sk(D)$ has the desired properties.

\begin{lemma}\label{L:length-skeleton}
 $\Sk(D)$ has length $O((g+t)|D|)$.
\end{lemma}
\begin{proof}
  Let $\gamma$ be a cycle of $\Gamma$.  By construction,  $\gamma^L$
  and~$\gamma^R$ have length at most
  $(1+\varepsilon)^{\theta(\gamma)+1}|\gamma^o|\le(1+\varepsilon)|\gamma|$.
  Moreover, if
  we added to~$\Sk(D)$ the projection of a shortest path connecting
  $\hat\gamma^L$ to~$\hat\gamma^R$, it is because $\hat\gamma$ had
  some edges leaving it on each on its sides.  In that case, since $D$
  is eligible, it has no degree one vertex and no face
  without terminal, and since the annulus~$\hat A_\gamma$ bounded by
  $\hat\gamma^L$ and~$\hat\gamma^R$ contains no terminal, it must be
  that $\hat D$ connects $\hat\gamma^L$ to~$\hat\gamma^R$.  So the
  shortest path between $\hat\gamma^L$ and~$\hat\gamma^R$ has length
  at most that of~$D$ (see Figure~\ref{F:skeleton1}, right).  Thus,
  each time we consider a cycle in the exhaustive family~$\Gamma$, the
  length of~$\Sk(D)$ increases by $O(|D|)$.  The desired result
  follows since $\Gamma$ has $O(g+t)$ cycles.
\end{proof}

\begin{lemma}\label{L:skeleton-forest}
  There is a graph~$D'$ drawn on~$S$ such that:
    \begin{itemize} 
    \item $D'$ has the same homotopy type as~$D$,
    \item $|D'|\leq(1+\varepsilon)|D|$,
    \item $\Sk(D)$ intersects every cycle of~$D'$,
    \item $D'$ is small.
    \end{itemize}
\end{lemma}

To prove this lemma, we will need the following lemma of Hass and Scott.

\begin{lemma}[{Hass and Scott~\cite[Theorem~2.4]{hs-ics-85}}]\label{L:hass-scott-1}
  Let $\gamma$ be a closed curve in general position on a surface with
  non-empty boundary, not simple but homotopic to a simple curve.  Then
  $\gamma$ forms a monogon or a bigon.
\end{lemma}

\begin{proof}[Proof of Lemma~\ref{L:skeleton-forest}]
  We start by providing an overview of the proof. 
  We will inductively build $D'$ from
  $D$ by pushing cycles of $D$ whenever $\Sk(D)$ does not
  meet them. If $D$ has a cycle~$\gamma$ that is disjoint
  from~$\Sk(D)$, then $\gamma$ must belong to the exhaustive
  family~$\Gamma$, and the shortest path between $\hat\gamma^L$
  and~$\hat\gamma^R$ has not been added.  So we can push $\hat\gamma$,
  together with the attached parts in~$\hat A_\gamma$, in the
  direction of these attached parts until it hits $\hat\gamma^R$.
  This does not increase too much the length, because $\hat\gamma^R$ is
  not much longer than~$\hat\gamma$, and preserves the fact that one has a
  multicut dual.

  Here are the details.  In an intermediate step, some parts of~$D'$
  and~$\Sk(D)$ will overlap, so that these two graphs will \emph{not} be in
  general position.  Initially, let $D':=D$.  Our construction is
  iterative, modifying~$D'$ as long as some cycle in~$D'$ does
  not intersect~$\Sk(D)$.  We maintain the following invariant: the part
  of~$D'$ outside~$\Sk(D)$ is included in~$D$.

  So, assume that there exists a cycle~$\gamma$ in~$D'$ that is
  disjoint from~$\Sk(D)$.  By the invariant above, $\gamma$ appears at the
  same place in~$D$, and is thus embedded.  We claim that $\gamma$ belongs to the exhaustive
  family~$\Gamma$ corresponding to~$D$.  Indeed, otherwise, by the
  definition of an exhaustive family, $\gamma$ and some
  cycle~$\gamma'\in\Gamma$ would be essentially crossing; but, by
  construction, $\Sk(D)$ contains some closed curves homotopic
  to~$\gamma'$, which would therefore cross~$\gamma$, which is a
  contradiction.  Thus $\gamma\in\Gamma$.  Moreover, $\Sk(D)$ contains a
  closed curve homotopic to~$\gamma$, which $\gamma$ would cross if it were
  one-sided; so $\gamma$ is two-sided.

  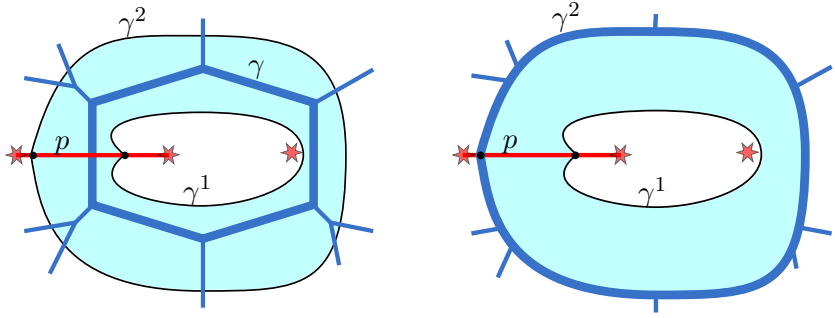
\begin{figure}[t]
    \begin{center}
      \def\svgwidth{11cm}
      \begingroup%
  \makeatletter%
  \providecommand\color[2][]{%
    \errmessage{(Inkscape) Color is used for the text in Inkscape, but the package 'color.sty' is not loaded}%
    \renewcommand\color[2][]{}%
  }%
  \providecommand\transparent[1]{%
    \errmessage{(Inkscape) Transparency is used (non-zero) for the text in Inkscape, but the package 'transparent.sty' is not loaded}%
    \renewcommand\transparent[1]{}%
  }%
  \providecommand\rotatebox[2]{#2}%
  \newcommand*\fsize{\dimexpr\f@size pt\relax}%
  \newcommand*\lineheight[1]{\fontsize{\fsize}{#1\fsize}\selectfont}%
  \ifx\svgwidth\undefined%
    \setlength{\unitlength}{1243.7654512bp}%
    \ifx\svgscale\undefined%
      \relax%
    \else%
      \setlength{\unitlength}{\unitlength * \real{\svgscale}}%
    \fi%
  \else%
    \setlength{\unitlength}{\svgwidth}%
  \fi%
  \global\let\svgwidth\undefined%
  \global\let\svgscale\undefined%
  \makeatother%
  \begin{picture}(1,0.36115679)%
    \lineheight{1}%
    \setlength\tabcolsep{0pt}%
    \put(0,0){\includegraphics[width=\unitlength,page=1]{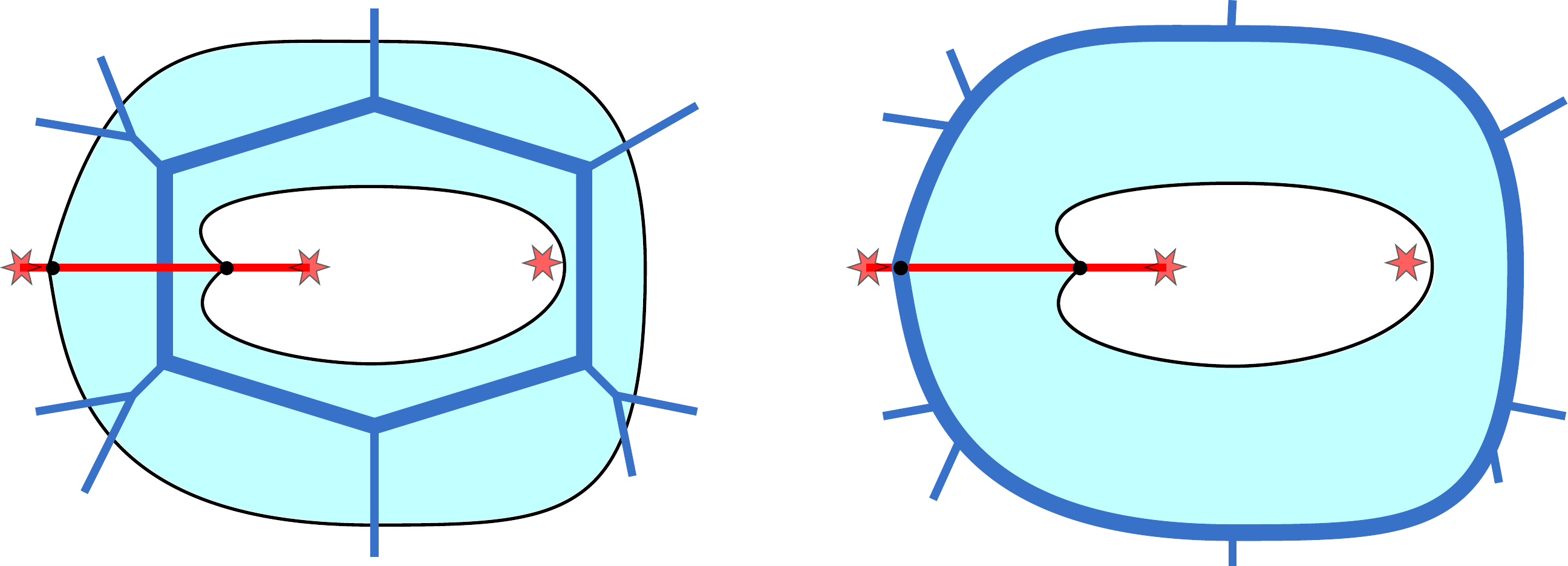}}%
    \put(0.21513442,0.13635269){\color[rgb]{0,0,0}\makebox(0,0)[lt]{\lineheight{0}\smash{\begin{tabular}[t]{l}$\gamma^L$\end{tabular}}}}%
    \put(0.13667407,0.34003789){\color[rgb]{0,0,0}\makebox(0,0)[lt]{\lineheight{0}\smash{\begin{tabular}[t]{l}$\gamma^R$\end{tabular}}}}%
    \put(0.29359482,0.28911659){\color[rgb]{0,0,0}\makebox(0,0)[lt]{\lineheight{0}\smash{\begin{tabular}[t]{l}$\gamma$\end{tabular}}}}%
    \put(0.05977254,0.19818572){\color[rgb]{0,0,0}\makebox(0,0)[lt]{\lineheight{0}\smash{\begin{tabular}[t]{l}$p$\end{tabular}}}}%
    \put(0.76502815,0.13470388){\color[rgb]{0,0,0}\makebox(0,0)[lt]{\lineheight{0}\smash{\begin{tabular}[t]{l}$\gamma^L$\end{tabular}}}}%
    \put(0.66044912,0.34600691){\color[rgb]{0,0,0}\makebox(0,0)[lt]{\lineheight{0}\smash{\begin{tabular}[t]{l}$\gamma^R$\end{tabular}}}}%
    \put(0.60145594,0.19852387){\color[rgb]{0,0,0}\makebox(0,0)[lt]{\lineheight{0}\smash{\begin{tabular}[t]{l}$p$\end{tabular}}}}%
  \end{picture}%
\endgroup%

\caption{In this simple case, we modify the multicut dual by
      pushing it to~$\gamma^R$, increasing its length by a factor of at
      most $1+\varepsilon$.}
       \label{F:skeleton2}
  \end{center}
  \end{figure}  

  Let $\hat\gamma$ be the lift of~$\gamma$ in the annular cover~$\hat
  S_\gamma$.  Since $\hat\gamma^L$ and $\hat\gamma^R$ belong to~$\Sk(D)$,
  they are not crossed by~$\hat\gamma$; by the choice of these closed
  curves, $\hat\gamma$ is entirely contained in the annulus~$\hat
  A_\gamma$.  Let $\hat D'_\gamma$ be the intersection of the interior
  of~$\hat A_\gamma$ with the connected component of the lift of~$D'$
  that contains~$\hat\gamma$.  Since $\gamma$ does not cross $\Sk(D)$, the
  shortest path between $\hat\gamma^L$ and $\hat\gamma^R$ was not added in
  the construction of~$\Sk(D)$.  This implies that there are edges adjacent to $\hat\gamma$ on at most of one its sides, and thus, since $\hat\gamma$ is separating in $\hat S_\gamma$, that $\hat D'_\gamma$ cannot
  touch both $\hat\gamma^L$ and~$\hat\gamma^R$. For example, let us assume
  that it does not touch $\hat\gamma^L$, the case where it does not touch
  $\hat\gamma^R$ being symmetric (if $\hat D'_\gamma$ touches neither $\hat\gamma^L$ nor $\hat\gamma^R$, we choose one arbitrarily).  We note that, by our invariant,
  the projection $D'_\gamma$ of~$\hat D'_\gamma$ to~$S$ is also
  part of~$D$, no face of which is a disk because $D$ is eligible.
  Thus, since $D$ has no degree zero or degree one vertex, $\hat D'_\gamma$ consists of the
  cycle~$\hat\gamma$ together with some disjoint trees attached to it and
  with leaves on~$\hat\gamma^R$.

  The na\"\i{}ve strategy now is to push $\hat D'_\gamma$
  onto~$\hat\gamma^R$, or more precisely to change $D'$ by removing
  $\hat D'_\gamma$ and by adding~$\hat\gamma^R$, as in
  Figure~\ref{F:skeleton2}.  However, there is a possible pitfall
  here, as two distinct points of $\hat D'_\gamma$ may project to the
  same point of~$S$ (see Figure~\ref{F:skeleton-pitfall}).  We therefore
  use a more complicated argument.  Let $\gamma^R$ be the projection
  of~$\hat\gamma^R$ on~$S$.  If $\gamma^R$ self-intersects, and since it
  is homotopic to the simple cycle~$\gamma$, it contains a monogon or a
  bigon (Lemma~\ref{L:hass-scott-1}).  We flip that monogon or bigon, so
  that the new curve has the same image as~$\gamma^R$ except in a small
  neighborhood of some self-crossing points, and has strictly less
  self-crossing points.  We repeatedly flip monogons and bigons in $\gamma$ until no more remain, obtaining a simple
  closed curve~$\gamma^{\prime R}$ homotopic to~$\gamma^R$ and with
  essentially the same image.  Now, $\gamma$ and~$\gamma^{\prime R}$ are
  simple, pairwise disjoint, homotopic in~$S$;
  furthermore, $\gamma$ cannot be contractible, since otherwise it would
  bound a disk~\cite[Theorem~1.7]{e-c2mi-66}.  It follows that $\gamma$ and $\gamma^{\prime R}$ bound an annulus~$A$ in~$S$~\cite[Lemma~2.4]{e-c2mi-66}.

  We claim that the edges of~$D$ incident to~$\gamma$, if there are any, leave~$\gamma$ on
  the side of~$A$.  Indeed, the pieces of~$\hat D'_\gamma$ belong to
  the annulus bounded by~$\hat\gamma$ and~$\hat\gamma^R$.  When we flip a
  monogon or bigon on~$\gamma^R$, this changes~$\hat\gamma^R$ to another
  closed curve that is disjoint from~$\hat\gamma$, and thus has to be on
  the same side of~$\hat\gamma$ in the annular cover (because flipping a
  monogon or bigon does not move the whole curve).  Thus, at the end of the
  process, the lift of the annulus~$A$ is on the same side of~$\hat\gamma$
  as was the annulus bounded by~$\hat\gamma$ and~$\hat\gamma^R$.

  We can now safely replace the part of~$D$ that belongs to~$A$ by the
  single cycle~$\gamma^{\prime R}$.  For the same reason as above, the part
  of~$D$ that belongs to~$A$ consists of the cycle~$\gamma$ together with
  disjoint trees with leaves on~$\gamma^{\prime R}$, so this operation can
  be represented as a sequence of edge compressions, edge expansions, and
  homotopies for~$D$.  This increases the length of~$D'$ by a factor of
  at most $1+\varepsilon$, because $\gamma^R$ (or $\gamma^{\prime R}$) is
  at most a factor $1+\varepsilon$ longer than~$\gamma$.  Finally, we
  slightly deform the part of $D'$ that overlaps $\gamma^{\prime R}$
  onto~$\gamma^R$, so that $D'$ really overlaps~$\Sk(D)$; this does
  not affect the length of~$D'$, and is possible because the image
  of~$\gamma^R$ is obtained from that of~$\gamma^{\prime R}$ by a small
  perturbation that makes pieces of~$\gamma^{\prime R}$ touch.

  We repeat this operation as long as there is a cycle~$\gamma$
  in~$D'$ that is not crossed by~$\Sk(D)$.  As proved above, such a cycle is
  necessarily part of~$\Gamma$, and after the iteration it is pushed
  onto~$\Sk(D)$, so this process eventually stops and the number of iterations
  is at most the number of cycles in~$\Gamma$.

  \begin{figure*}[t]
    \begin{center}
      \def\svgwidth{\linewidth}
      \begingroup%
  \makeatletter%
  \providecommand\color[2][]{%
    \errmessage{(Inkscape) Color is used for the text in Inkscape, but the package 'color.sty' is not loaded}%
    \renewcommand\color[2][]{}%
  }%
  \providecommand\transparent[1]{%
    \errmessage{(Inkscape) Transparency is used (non-zero) for the text in Inkscape, but the package 'transparent.sty' is not loaded}%
    \renewcommand\transparent[1]{}%
  }%
  \providecommand\rotatebox[2]{#2}%
  \newcommand*\fsize{\dimexpr\f@size pt\relax}%
  \newcommand*\lineheight[1]{\fontsize{\fsize}{#1\fsize}\selectfont}%
  \ifx\svgwidth\undefined%
    \setlength{\unitlength}{467.2121025bp}%
    \ifx\svgscale\undefined%
      \relax%
    \else%
      \setlength{\unitlength}{\unitlength * \real{\svgscale}}%
    \fi%
  \else%
    \setlength{\unitlength}{\svgwidth}%
  \fi%
  \global\let\svgwidth\undefined%
  \global\let\svgscale\undefined%
  \makeatother%
  \begin{picture}(1,0.1773203)%
    \lineheight{1}%
    \setlength\tabcolsep{0pt}%
    \put(0,0){\includegraphics[width=\unitlength,page=1]{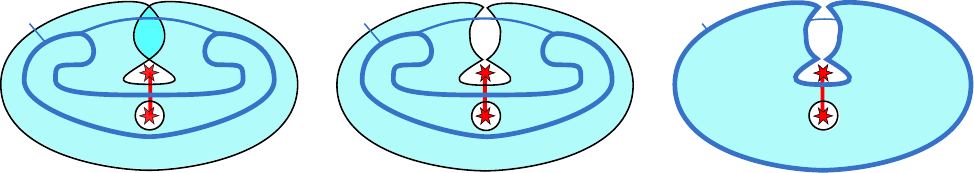}}%
    \put(0.03988074,0.05256574){\color[rgb]{0,0,0}\makebox(0,0)[lt]{\lineheight{0}\smash{\begin{tabular}[t]{l}$\gamma$\end{tabular}}}}%
    \put(0.26773669,0.15635353){\color[rgb]{0,0,0}\makebox(0,0)[lt]{\lineheight{0}\smash{\begin{tabular}[t]{l}$\gamma^R$\end{tabular}}}}%
    \put(0.17030639,0.05744702){\color[rgb]{0,0,0}\makebox(0,0)[lt]{\lineheight{0}\smash{\begin{tabular}[t]{l}$\gamma^L$\end{tabular}}}}%
    \put(0.51745862,0.05744702){\color[rgb]{0,0,0}\makebox(0,0)[lt]{\lineheight{0}\smash{\begin{tabular}[t]{l}$\gamma^L$\end{tabular}}}}%
    \put(0.38524355,0.05229412){\color[rgb]{0,0,0}\makebox(0,0)[lt]{\lineheight{0}\smash{\begin{tabular}[t]{l}$\gamma$\end{tabular}}}}%
    \put(0.61309951,0.1560819){\color[rgb]{0,0,0}\makebox(0,0)[lt]{\lineheight{0}\smash{\begin{tabular}[t]{l}$\gamma^{\prime R}$\end{tabular}}}}%
  \end{picture}%
\endgroup%

      \caption{Left: An example where $\gamma^R$ is non-simple (although
        its lift $\hat\gamma^R$ always is).  In that example, some points
        on the thin horizontal edge lifts to two distinct points
        in~$\hat A_\gamma$.  We iteratively flip the monogons and bigons
        of~$\gamma^R$ to make it simple (middle), and can then replace
        $\gamma$ with the new $\gamma^R$ (right).}
         \label{F:skeleton-pitfall}
    \end{center}
  \end{figure*}
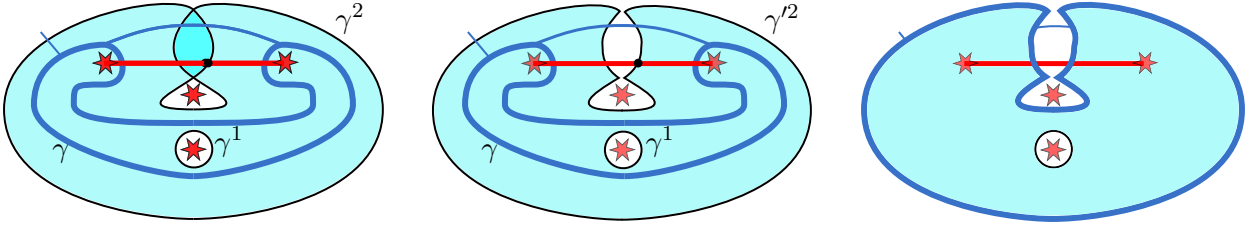
  
  Finally, we slightly perturb $D'$ so that $D'$ and~$\Sk(D)$ are in
  general position to remove the overlaps.  We do so by ensuring that
  each edge of~$D'$ intersecting~$\Sk(D)$ before the perturbation
  still intersects it after the perturbation.  This preserves the fact
  that $\Sk(D)$ intersects every cycle in~$D'$.  This concludes
  the definition of~$D'$, which by construction has the same homotopy type as~$D$.
  
  At each of the
  iterations above, the length of~$D'$ increases by at most
  $(1+\varepsilon)$ times the length of the cycle~$\gamma$ considered
  in that iteration.  But the cycles~$\gamma$ considered are
  edge-disjoint, because once a cycle~$\gamma$ is considered, it is
  contained in~$\Sk(D)$.  So, after all the iterations, the length
  of~$D'$ has increased by a factor of at most $(1+\varepsilon)$.
  Hence the length of~$D'$ is at most $(1+\varepsilon)$ times the length of~$D$, as
  desired.

  There remains to prove that $D'$ is small.  Conditions \ref{sm:genpos},
  \ref{sm:nodeg01}, and~\ref{sm:compl} are obviously satisfied.  To prove
  Condition~\ref{sm:cross}, it suffices to prove that each lift~$\hat q$ of
  an arc of~$K$ in~$\hat S_\gamma$ is not crossed more by~$\hat\gamma^R$
  than by~$\hat\gamma$. This follows from the fact that arcs of $K$ are
  shortest homotopic paths and from the isolation rule for shortest paths
  explained in Section~\ref{S:topology}.
\end{proof}

\begin{proof}[Proof of Proposition~\ref{P:oneskeleton}]
  The result now immediately follows from Lemmas~\ref{L:numskeleta},
  \ref{L:build-skeleton}, \ref{L:length-skeleton},
  and~\ref{L:skeleton-forest}, observing that $\Sk(D)$ has $O(g+t)$
  vertices and edges by construction.
\end{proof}

\section{Near-optimal solution made of trees and cycles}
\label{S:portals}
Once the skeleta have been computed, one can compute a family of
\emph{portal sets}, one portal set per skeleton.

\begin{proposition}\label{P:portals}
  In $((g+t)/\varepsilon)^{O(g+t)}\cdot n\log n$ time, we can compute
  a family of
  $k=((g+t)/\varepsilon)^{O(g+t)}$ \emphdef{portal sets}
  $\{\PP_1, \ldots \PP_k\}$ such that each set consists of
  $O((g+t)^2/\varepsilon)$ points on~$S$, called \emphdef{portals}, and one
  of the portal sets, denoted by $\PP$, satisfies the following condition:
  There exists a good multicut dual $C''$ such that $C''_0\setminus\PP$ is
  a forest with $O(g+t)$ leaves.
\end{proposition}               %
\begin{proof} 
We describe how to construct the sets 
  $\PP_1,\ldots, \PP_k$.
  There is a set $\PP_i$ for each skeleton obtained by
  Proposition~\ref{P:skeleton}.
  Fix a skeleton $\Sigma$ and let $e$ be an edge of $\Sigma$.
  Locate portals on~$e$ in a way that each
  point on~$e$ is at distance $O(\varepsilon |\Sigma|/(g+t)^2)$ of some
  portal, where as usual the $O(\cdot)$ notation hides a universal
  constant. This requires at most $O(|e|(g+t)^2/(\varepsilon |\Sigma|))$ portals.  Doing
  this for each of the $O(g+t)$ edges of~$\Sigma$ produces $O((g+t)^2/\varepsilon)$
  portals; these constitute the set $\PP_i$ corresponding
  to $\Sigma$. Doing this for each of the $((g+t)/\varepsilon)^{O(g+t)}$ skeleta
  yields the portal sets $\{\PP_1, \ldots \PP_k\}$.
  Recall that the algorithm of Section~\ref{S:skeleton} 
  encodes each skeleton by its abstract graph and,
  for each of its edges $e$, by the ordered sequence of the edges of $G$
  crossed by $e$ in the drawing, together with the orientation of each
  crossing. Thus, the algorithm has enough
  information to define the portal sets $\{\PP_1, \ldots \PP_k\}$.

  Now, following Proposition~\ref{P:skeleton}, let $\Sk(C_0)$ be a skeleton
  of length $O((g+t)\OPT)$ and $C'$ be the corresponding good multicut dual, 
  such that $C'_0\setminus\Sk(C_0)$ is a forest. This
  skeleton corresponds to one of the sets of portals, which we denote
  by $\PP$. 

  Consider the following subgraph $C''_0$ obtained
  by modifying $C'_0$ as follows:
  For every edge $e$ of $C'_0$ crossed by $\Sk(C_0)$, we insert, on~$e$, a
  \emph{single} detour  to the closest portal in $\PP$ via a shortest
  path.  Let $C''_0$ be the resulting graph.  By construction, each detour
  has length $O(\varepsilon |\Sk(C_0)|/(g+t)^2)$.  Moreover,
  $C''_0\setminus\PP$ is a forest with $O(g+t)$ leaves, because $C''_0$ has
  $O(g+t)$ edges, and each edge passes at most once through~$\PP$.  
  We now prove that $C''=(C''_0, C'_1, \ldots C'_k)$ is a good multicut dual.

  Conditions~\ref{c:small}\ref{sm:genpos}, \ref{c:small}\ref{sm:nodeg01},
  and~\ref{c:small}\ref{sm:compl} are satisfied by construction.  Since
  $C'_0$ crosses each arc of $K$ at most $O(g+t)$ times, each of the
  $O(g+t)$ detours consists of two shortest paths, and every arc in $K$ is
  the concatenation of two shortest paths, the number of crossings between
  $C''_0$ and each arc of~$K$ is $O(g+t)$, which gives
  Condition~\ref{c:small}\ref{sm:cross}.  Hence $C''$ is small.  Conditions
  \ref{c:cycles} and~\ref{c:nbcycles} are satisfied by construction.  The
  number of detours made is $O(g+t)$, and each of them has length
  $O(\varepsilon |\Sk(C_0)|/(g+t)^2)=O(\varepsilon\OPT/(g+t))$, so the
  length of $C''_0$ is at most $(1+O(\varepsilon))\OPT$, and
  Condition~\ref{c:length} is fulfilled.  Finally, $C''_0$ is obtained from
  $C'_0$ by a homotopy of the edges, so Condition~\ref{c:homology} is
  fulfilled.  Therefore, $C''$ is a good multicut dual.

  For each skeleton $\Sigma$ constructed by \textsc{BuildAllSkeleta}, the
  portals can be placed greedily by ``walking'' along each edge of
  $\Sigma$, in time linear in the complexity of $\Sigma$, which is
  $O((g+t)^3n)$. This yields an overall complexity of
  $((g+t)/\eps)^{O(g+t)}\cdot n \log n$.
\end{proof}

\section{The algorithm}
\label{S:alg}
The final step of our algorithm is to compute separately, for each set of
portals $\PP_i$ and for each possible position of a good multicut dual with
respect to $\PP_i$ and $K$, the optimal layout of this multicut dual. By
Proposition~\ref{P:portals}, there exists a good multicut dual $C''$ such
that $C_0''$ is cut into a forest by a set of portals $\PP$. Then, finding
$C_0''$ is achieved by (1) guessing the topology of the trees of the forest
and (2) for each of these, computing a Steiner tree corresponding to this
topology in some portion of the universal cover of $S$. The cycles $C''_i$
for $i\ge1$ are found by computing shortest homotopic closed
curves~\cite{ce-tnpcs-10} (see also Section~\ref{S:annular}).

Throughout this section, we consider homotopies of trees as defined in Section~\ref{S:homotopy-type}, where the leaves are fixed. The \emphdef{layout} of a forest $F$ drawn on $S$ is the collection
of homotopy types of the trees in $F$. These trees can be lifted to the universal
cover and, by Lemma~\ref{L:trees},
specifying a layout is the same as specifying the sets of lifts of
portals corresponding to each tree (after fixing a basepoint). It is a \emphdef{good layout} if
the forest $F$ has $O(g+t)$ leaves and if $F$ and each arc of~$K$ cross $O(g+t)$ times.

\begin{lemma}\label{L:homotopictrees}
Let $F=(T_1, \ldots, T_m)$ be a forest drawn on $S$ with a good layout. We
can compute in time $(g+t)^{O((g+t)^2)}n \log n$ a shortest forest
homotopic to $F$, i.e., a family of trees $(T'_1, \ldots , T'_m)$ drawn on
$S$ such that for each $i$, $T'_i$ is a shortest tree homotopic to $T_i$.
\end{lemma}

The proof of this lemma relies on the following theorem, which follows from the
classical algorithm of Dreyfus and Wagner~\cite{dw-spg-71}, sped up
using the techniques of Erickson, Monma, and
Veinott~\cite{emv-ssmmc-87} (this is not explicitly stated in that article, but it is folklore that the method applies to this problem, see, e.g., Vygen~\cite[Introduction]{v-faost-11}):
\begin{theorem} \label{T:steiner} Given a planar graph $G=(V,E)$ with $n$
  vertices and a set $T\subseteq V$ of $t$ terminals, one can compute an
  optimal Steiner tree of~$G$ with respect to~$T$ in time
  $2^{O(t)} n \log n$.
\end{theorem}
\noindent (We remark that the planar Steiner Tree PTAS of Borradaile, Klein, and
Mathieu~\cite{bkm-onasst-09} is exponential in $1/\varepsilon$, and thus
too slow for our purposes.)

\begin{proof}[Proof of Lemma~\ref{L:homotopictrees}]
Similarly to the computation of shortest homotopic paths described in
Section~\ref{S:universal}, the appropriate setting to compute Steiner
trees in our algorithm is a relevant region of the universal cover. A tree
$T$ drawn on $S$ connecting a set of leaves $L$ can be lifted in the
universal cover of $S$ to a tree $\tilde T$ connecting some lifts of the
leaves $\tilde L$. Computing a Steiner tree on the points $\tilde L$ and
projecting it back to $S$ yields the shortest tree with the same homotopy
type as $T$ connecting the leaves $\tilde L$. Indeed, after choosing a basepoint for
the lifts, every tree with the same homotopy type lifts to a tree that
connects the same lifts, by the definition of the universal
cover. Furthermore, it is enough to compute a Steiner tree in the relevant
region of $\tilde{S}$, defined as follows: If $D$ is the disk obtained by
cutting $S$ along~$K$, the relevant region is the union of the lifts of~$D$ visited
by~$\tilde T$.  The claim is that a lift of a shortest tree with the same
homotopy type as~$T$ belongs to the relevant region; as for shortest
homotopic paths, this follows from the fact that arcs in $K$ are shortest
homotopic arcs.

A relevant region
of the universal cover for \textit{all} the possible good layouts has
size $(g+t)^{O((g+t)^2)}n$: Since there are $O((g+t)^2)$ crossings
between the edges of $C''_0$ and the arcs of $K$, one can compute the relevant region
by cutting $S$ along $K$, yielding a disk $D$, and pasting new copies
of $D$ at its boundaries in a tree-like fashion up to depth
$O((g+t)^2)$. Since there are $O(g+t)$ arcs in $K$, the corresponding
relevant region has size $(g+t)^{O((g+t)^2)}n$.

In order to compute a shortest forest homotopic to $F$, we first
compute the relevant region. It is a planar graph, and we use
Theorem~\ref{T:steiner} to compute Steiner trees for each of lifts of
the trees in $F$. This is done in time
$2^{O(g+t)}(g+t)^{(O(g+t)^2)}n \log n=(g+t)^{(O(g+t)^2)}n \log
n$. Projecting the union of these trees back to the surface gives the
shortest forest homotopic to $F$.
\end{proof}

Similarly, for a closed curve $\gamma$ drawn on $S$, its homotopy class is \emphdef{good}
if the shortest homotopic closed curve crosses each arc of $K$ at most $O(g+t)$ times.

We now have all the tools to describe our algorithm. 

\bigskip
\textsc{Main algorithm}
\begin{enumerate}
\item Compute the family of skeleta using the \textsc{BuildAllSkeleta} algorithm of Section~\ref{S:skeleton}.
\item For each skeleton, compute a set of portals $\PP_i$ following Section~\ref{S:portals}.
\item For every set of portals $\PP_i$ and for every good layout of a
  forest with leaves on~$\PP_i$, compute a shortest homotopic forest corresponding to this layout using Lemma~\ref{L:homotopictrees}.
\item\label{c:homotopy} For every choice of good homotopy class of $j$ closed curves, for every
  integer $0 \leq j \leq O(g+t)$, compute shortest homotopic closed curves. All the possible combinations of these
  additional closed curves with the previous family of graphs yield the candidate solutions.
\item Output the shortest of the candidate solutions that is a multicut dual.
\end{enumerate}

\bigskip

\noindent We remark that, in step~\ref{c:homotopy}, we could rely on shortest \textit{homologous} closed curves instead of shortest homotopic closed curves. This follows from the homological nature of the multicut dual property. Those would be easier to enumerate and compute (see Erickson and Nayyeri~\cite{en-mcsnc-11}) but would not improve the overall complexity since this is not the bottleneck of the algorithm.

\begin{proof}[Proof of Theorem~\ref{T:main}]
\textbf{Running time.}
The computations in Sections~\ref{S:skeleton} and~\ref{S:portals} have
already been shown to take time $((g+t)/\varepsilon)^{O(g+t))}n \log
n$.

Each portal set~$\PP_i$ has $O((g+t)^2/\varepsilon)$ portals, which lift to
$O((g+t)^2/\varepsilon) \cdot
(g+t)^{O(g+t)^2}=(g+t)^{O(g+t)^2}/\varepsilon$ points in the relevant
region of the universal cover.  To specify a good layout, we need to choose
a set~$\hat{P}$ of $O(g+t)$ such lifted portals, which are the leaves of
the forest; this gives $(1/\eps)^{O(g+t)}\cdot(g+t)^{O(g+t)^3}$
possibilities.  Then we choose the number
$m=O(g+t)$ of trees and pick for each tree a subset of lifted
portals from the set $\hat{P}$; there are
$2^{O(g+t)}$ choices of subset for each tree. Thus the number of good
layouts of forests is
$1/\eps^{O(g+t)} \cdot (g+t)^{O((g+t)^3)}$.

A good homotopy class for a closed curve is characterized by the sequence of
crossings with $K$, and thus there are $(g+t)^{O((g+t)^2)}$ good
homotopy classes. This gives $(g+t)^{O(g+t)^3}$ choices
for the set of additional closed curves in the main algorithm.

Now, for each good layout of a forest, we compute the shortest homotopic forest using Lemma~\ref{L:homotopictrees}; this takes time $(g+t)^{O((g+t)^2)}n \log n$. Similarly, for each good layout of a closed curve, we compute the
shortest homotopic closed curve in the relevant region of its annular cover,
which is done using the tools exposed in
Section~\ref{S:annular}. Since the relevant region of the annular
cover has size $O((g+t)n)$, these computations are negligible compared
to the other ones.

Hence, the algorithm has an overall running time of
$1/\varepsilon^{O(g+t)}(g+t)^{O((g+t)^3)} \cdot n \log n$.

\bigskip
\textbf{Correctness and approximation guarantee.}
By Proposition~\ref{P:portals}, there is a good multicut dual $C''=(C''_0,C''_1,\ldots,C''_k)$ and a
set of portals $\PP$ such that $C''_0\setminus \PP$ is a forest with
$O(g+t)$ leaves.   By definition of a good multicut dual, $C''_0\setminus
P$ is a good layout of a forest with respect to~$\PP$, and $C''_1,\ldots,C''_k$ are good layouts
of closed curves. For this choice
of portals and layouts, our algorithm will output a graph $C'''$ made of a
set of closed curves $C'''_1,\ldots,C'''_k$, and a graph $C'''_0$ (possibly non-connected) such that:
\begin{itemize}
\item $C'''_0$ has the same homotopy type as $C''_0$, since the trees it consists of have the same homotopy type as the ones of $C''_0$;
\item $C'''_0$ is no longer than $C''_0$, since $C'''_0$ is made of Steiner trees 
  that are by definition no longer than the corresponding trees in
  $C''_0$;
\item each $C'''_i$, $1\le i\le k$, is homotopic to $C''_i$, since they have the same layout, and each $C'''_i$ is shorter than $C''_i$ by construction.
\end{itemize}

Therefore, by Lemma~\ref{L:preserve}, $C'''$ has the multicut dual
property, and thus by
Lemma~\ref{L:homology}, it is a multicut dual. Hence the graph output
by our algorithm is a multicut dual of length at most the length of $C''$,
which concludes the proof since $C''$ is a good multicut dual and thus a
near-optimal solution.
\end{proof}

\section*{Acknowledgments}

We would like to thank the referees for their numerous observations and
suggestions, which greatly improved the presentation of the article.

\bibliographystyle{siam}
\bibliography{biblio}
\appendix
\section{Covering spaces and shortest homotopic curves}
\label{A:covering}
In this appendix, we introduce the covering spaces that are used throughout this article, and summarize the most important points of the algorithms to compute shortest homotopic paths and closed curves~\cite[Sections 1
and~6]{ce-tnpcs-10}, eluding many technicalities
from that paper that are irrelevant to our problem.


A map $\pi:S' \to S$ between two surfaces is called a \emphdef{covering
  map} if it is a local homeomorphism, or more precisely if each point
$x\in S$ lies in an open neighborhood $U$ such that (1)~$\pi^{-1}(U)$ is a
countable union of disjoint open sets $U_1\cup U_2\cup\cdots$ and (2)~for
each~$i$, the restriction $\pi|_{U_i}:U_i \to U$ is a homeomorphism.  We
say that $S'$ (more formally, it should be $(S',\pi)$) is a
\emphdef{covering space} of $S$.  The definition implies that every
path~$p$ in~$S$ can be \emphdef{lifted} to a path~$p'$ in~$S'$ (the
opposite operation of \emphdef{projecting} with the map~$\pi$), and that
moreover this can be done in a unique way if the starting point~$x'$
of~$p'$ is prescribed (such that $\pi(x')$ equals the initial point
of~$p$).

An important type of covering space is the \emphdef{universal cover}
of~$S$, denoted by~$\tilde S$, in which every closed curve is
contractible.  The universal cover is unique; a closed curve in~$S$ is
contractible if and only if its lifts in~$\tilde S$ are closed curves.


\subsection{Universal cover and shortest homotopic paths}\label{S:universal}

Let $p$ be a path on~$S$.  By definition of the universal cover, to compute
a shortest path homotopic to~$p$, it suffices to (1) compute a
lift~$\tilde p$ of~$p$ in the universal cover~$\tilde S$, which is
naturally a (non-compact) cross-metric surface, (2) compute a shortest path
between the endpoints of~$\tilde p$, and (3) return the projection of the
resulting path to~$S$.  However, the universal cover is infinite, and one
needs to restrict oneself to a finite portion of the universal cover, its
\emph{relevant region} with respect to~$\tilde p$.  Let us explain in more
detail.

Let $D$ be the disk obtained by cutting $S$ along~$K$; this is naturally a
cross-metric surface defined by the image of~$G$.  The universal cover~$S$
can be built by gluing together infinitely many copies of~$D$ along lifts
of the arcs in~$K$; the \emphdef{relevant region} of~$\tilde S$ with
respect to~$\tilde p$ is defined to be the set of lifts of~$D$ visited
by~$\tilde p$.  The claim is that a shortest path in~$\tilde S$ connecting
the endpoints of~$\tilde p$ stays in the relevant region, and this follows
from the fact that, since the arcs in~$K$ are shortest homotopic arcs, they
lift to shortest paths, and each of these shortest paths separate the
universal cover.

We can easily build the relevant region of the universal cover
incrementally: Initially, it is the single copy of~$D$ containing the
source of~$\tilde p$.  We walk along~$\tilde p$ in the relevant region, and
each time we are about to exit the relevant region via a lift~$\tilde a$ of
an arc of~$K$, we attach a copy of~$D$ across~$\tilde a$, preserving the
local homeomorphism condition.  This results in a set of copies of~$D$
attached together in a tree-like fashion.

From the above discussion, it also follows that a shortest path~$\tilde p'$
with the same endpoints as~$\tilde p$ crosses each lift of an arc in~$K$ at
most once, and crosses only lifts of arcs of~$K$ that are already crossed
by~$\tilde p$.  This is because each lift of~$K$ is a shortest path; to
avoid $\tilde p'$ crossing twice the same lift of~$K$, as mentioned in Section~\ref{S:topology}, we use the
number of crossings with~$K$ as a tie-breaking measure for the notion of
length.  In particular, the projection~$p'$ of~$p$ crosses each arc of~$K$
at most as many times as $p$ does.

Finally, since computing shortest paths in a planar graph takes
linear time~\cite{hkrs-fspap-97}, and since the relevant region has
complexity $O((g+t)kn)$ where $k$ is the number of crossings between $p$
and~$K$, computing a shortest homotopic path can be done in time
$O((g+t)kn)$.


\subsection{Annular covers and shortest homotopic closed curves}\label{S:annular}
The same strategy can be used for computing a shortest homotopic closed
curve homotopic to a non-contractible closed curve~$\gamma$ in~$S$, but now
we have to use the \emphdef{annular cover} of~$S$ with respect to~$\gamma$,
denoted by~$\hat S_\gamma$; this is a covering space in which every simple
closed curve is either contractible or homotopic to a lift of~$\gamma$ or
to its reverse.  If $\gamma$ is two-sided, the annular cover is
homeomorphic to an annulus with some points of the boundary removed.
Otherwise, $\gamma$ is one-sided, and the annular cover is homeomorphic to
a M\"obius strip with some points of the boundary removed.  The shortest
homotopic closed curve in the annular cover crosses every lift of an arc
in~$K$ at most once, and crosses only arcs that are crossed by the
lift~$\hat\gamma$ of~$\gamma$ that is a closed curve (so, as above with
paths, a shortest homotopic closed curve crosses each arc at most as many
times as the original closed curve does).  In the annulus case, the proof
is well known, essentially as before, and in the case where the annular
cover is a M\"obius strip, it rests on the following lemma.
\begin{lemma}\label{L:mobius-aux}
  Let $M$ be a cross-metric surface that is a M\"obius strip.  Let $a$ be a
  non-separating arc that is as short as possible in its homotopy class.
  Then some shortest non-contractible closed curve in~$M$ crosses~$a$
  exactly once.
\end{lemma}

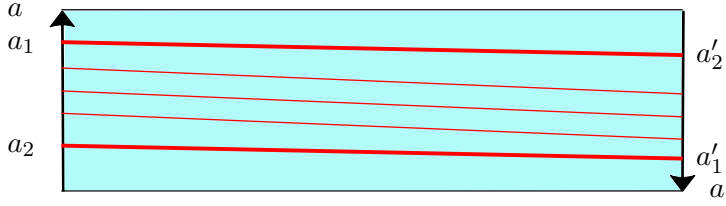
\begin{figure}
  \begin{center}
    \def\svgwidth{10cm}
    \begingroup%
  \makeatletter%
  \providecommand\color[2][]{%
    \errmessage{(Inkscape) Color is used for the text in Inkscape, but the package 'color.sty' is not loaded}%
    \renewcommand\color[2][]{}%
  }%
  \providecommand\transparent[1]{%
    \errmessage{(Inkscape) Transparency is used (non-zero) for the text in Inkscape, but the package 'transparent.sty' is not loaded}%
    \renewcommand\transparent[1]{}%
  }%
  \providecommand\rotatebox[2]{#2}%
  \newcommand*\fsize{\dimexpr\f@size pt\relax}%
  \newcommand*\lineheight[1]{\fontsize{\fsize}{#1\fsize}\selectfont}%
  \ifx\svgwidth\undefined%
    \setlength{\unitlength}{225.07965465bp}%
    \ifx\svgscale\undefined%
      \relax%
    \else%
      \setlength{\unitlength}{\unitlength * \real{\svgscale}}%
    \fi%
  \else%
    \setlength{\unitlength}{\svgwidth}%
  \fi%
  \global\let\svgwidth\undefined%
  \global\let\svgscale\undefined%
  \makeatother%
  \begin{picture}(1,0.25266164)%
    \lineheight{1}%
    \setlength\tabcolsep{0pt}%
    \put(0,0){\includegraphics[width=\unitlength,page=1]{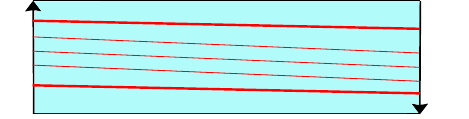}}%
    \put(-0.00223546,0.19982622){\color[rgb]{0,0,0}\makebox(0,0)[lt]{\lineheight{0}\smash{\begin{tabular}[t]{l}$a_1$\end{tabular}}}}%
    \put(0.91340881,0.04492524){\color[rgb]{0,0,0}\makebox(0,0)[lt]{\lineheight{0}\smash{\begin{tabular}[t]{l}$a'_1$\end{tabular}}}}%
    \put(0.91340881,0.18261499){\color[rgb]{0,0,0}\makebox(0,0)[lt]{\lineheight{0}\smash{\begin{tabular}[t]{l}$a'_2$\end{tabular}}}}%
    \put(-0.00223546,0.06213646){\color[rgb]{0,0,0}\makebox(0,0)[lt]{\lineheight{0}\smash{\begin{tabular}[t]{l}$a_2$\end{tabular}}}}%
    \put(-0.00223546,0.24285427){\color[rgb]{0,0,0}\makebox(0,0)[lt]{\lineheight{0}\smash{\begin{tabular}[t]{l}$a$\end{tabular}}}}%
    \put(0.93136262,0.00189718){\color[rgb]{0,0,0}\makebox(0,0)[lt]{\lineheight{0}\smash{\begin{tabular}[t]{l}$a$\end{tabular}}}}%
  \end{picture}%
\endgroup%

    \caption{Proof of Lemma~\ref{L:mobius-aux}: The shortest non-contractible
      closed curve~$\gamma$ in the M\"obius strip~$M$ corresponds, in
      this representation, to pairwise disjoint simple arcs connecting
      the left and right sides of the quadrangle.  If $\gamma$ contains
      more than one arc, as is depicted, then $a_1$ must be identified
      to $a'_1$ and similarly $a_2$ must be identified with $a'_2$.  This
      implies that $\gamma$ crosses~$a$ exactly twice, which in turn
      implies that it is contractible, a contradiction.}
    \label{F:mobius}
  \end{center}
\end{figure}

\begin{proof}
  Let $\gamma$ be a shortest non-contractible closed curve in~$M$,
  crossing~$a$ as few times as possible; $\gamma$ is simple and must
  cross~$a$ at least once (for otherwise it would be contractible in~$M$).
  We now prove that it crosses~$a$ exactly once; see Figure~\ref{F:mobius}.
  Cut~$M$ along~$a$, obtaining a quadrangle~$Q$ with four sides, two
  opposite sides (say the left and right sides) corresponding to the sides
  of~$a$.  Because $a$ is a shortest homotopic arc, and because the number
  of crossings with~$a$ is minimal, the image of~$\gamma$ in~$Q$ is made of
  simple disjoint arcs connecting the opposite sides of~$Q$ corresponding
  to~$a$.  Assume that there are at least two arcs.  Because the arcs are
  simple and disjoint, the left endpoint of the topmost arc, $a_1$, must be
  identified to the right endpoint of the bottommost arc, $a'_1$, and
  similarly the left endpoint of the bottommost arc, $a_2$, must be
  identified to the right endpoint of the topmost arc, $a'_2$, which
  implies that there are exactly two arcs, implying that $\gamma$
  crosses~$a$ exactly twice.  This is impossible since a closed curve is
  non-contractible if and only if it crosses $\gamma$ an odd number of
  times.
\end{proof}

A variation of the construction described in Section~\ref{S:universal}
allows to build the \emphdef{relevant region} of the annular
cover~$\hat S_\gamma$, which is topologically an annulus (in the two-sided
case) or a M\"obius strip (in the one-sided case) made of finitely many
copies of~$D$.  In essence, we first start at an arbitrary point
of~$\gamma$ and build the relevant region of the closed path starting at
that point.  This is a topological disk obtained by gluing together copies
of~$D$ in a tree-like fashion.  We only keep the copies of~$D$ that are
between the starting and ending points of the lift, discarding the other
copies, obtaining a set of copies of~$D$ glued together along a path.
Finally, we identify the initial and final copies of~$D$ in this path,
leading to an annulus or a M\"obius strip.

Algorithmically, we need to compute a shortest non-contractible closed
curve in an annulus or a M\"obius strip of complexity $O((g+t)kn)$, where
$k$ is the number of crossings between $\gamma$ and~$K$.  We can do this in
$O((g+t)kn\log((g+t)kn))$ time.
We describe a possible algorithm for the case of the annulus, the ideas of
which will be reused in Section~\ref{S:skeleton} (the result is actually
well-known, because it reduces to a minimum cut computation in the dual,
which can even be done in $O((g+t)kn\log\log((g+t)kn))$
time~\cite{insw-iamcm-11}).
\begin{lemma}\label{L:annulus}
  Let $A$ be a cross-metric annulus with complexity~$m$.  In
  $O(m\log m)$ time, we can compute a shortest non-contractible closed
  curve in~$A$.
\end{lemma}
\begin{proof}
  Compute a shortest path~$a$ between the two boundaries of~$A$.  It is
  easy to see that a desired shortest non-contractible closed crosses~$a$
  exactly once.  Cut~$A$ along~$a$, obtaining a quadrangle~$Q$ with four
  sides, two opposite sides (say the left and right sides) corresponding to
  the sides of~$a$.

  The shortest non-contractible closed curve in~$A$ corresponds to a
  shortest path between matching pairs of points in~$Q$.  We can compute
  the distance between all the matching pairs in $O(m\log m)$ time using
  the planar multiple-source shortest path algorithm by
  Klein~\cite{k-msppg-05}; the shortest non-contractible closed curve
  in~$A$ connects the matching pair with the smallest distance.
\end{proof}

In the case of the M\"obius strip, we expect this to be known in some
circles, but to our knowledge the result has not appeared anywhere:
\begin{lemma}\label{L:mobius}
  Let $M$ be a cross-metric M\"obius strip with complexity~$m$.  In
  $O(m\log m)$ time, we can compute a shortest non-contractible closed
  curve in~$M$.
\end{lemma}
\begin{proof}
  Attach a disk to the boundary component of~$M$, obtaining a projective
  plane~$P$, inheriting the cross-metric structure from~$M$.  In~$P$,
  compute a shortest non-contractible loop based at some point inside the
  disk~\cite[Lemma~5.2]{eh-ocsd-04}; that loop is simple and crosses the
  boundary of~$M$ exactly twice.  Its trace on~$M$ is a shortest arc~$a$
  among all non-separating arcs in~$M$.  By Lemma~\ref{L:mobius-aux}, some
  shortest non-contractible closed curve in~$M$ crosses~$a$ exactly once,
  and thus corresponds to a shortest path between matching pairs of points
  in the disk obtained by cutting $M$ along~$a$, as in the annulus case; we
  can similarly conclude using Klein's algorithm~\cite{k-msppg-05}.
\end{proof}

\end{document}